\documentclass[journal]{IEEEtran}
\usepackage[pdftex]{graphicx,color}
\usepackage{curves,epic,latexsym,amsmath,amssymb,array,cite,ifthen,url}
\newcommand{\pos}[2]{\makebox(0,0)[#1]{#2}}
\newcommand{\eqdef}{\stackrel{\scriptscriptstyle\bigtriangleup}{=} }

\newcommand{\Z}{\mathbb{Z}}

\newcommand{\calB}{\mathcal{B}}

\newcommand{\calE}{\mathcal{E}}
\newcommand{\calL}{\mathcal{L}}

\newcommand{\maxD}{{\hat D}}

\newcommand{\rank}{\operatorname{rank}}

\newcommand{\argmax}{\operatornamewithlimits{argmax}}
\newcommand{\lcf}{\operatorname{lcf}}

\newcommand{\lcm}{\operatorname{lcm}}

\newcommand{\rd}{\operatorname{rd}}

\newcommand{\odiv}{\,\operatorname{div}}

\definecolor{grey}{rgb}{0.5, 0.5, 0.5}
\definecolor{purple}{rgb}{1, 0, 1}

\newcommand{\grey}{\color{grey}}

\newcounter{examplecntr}
{\begin{trivlist}\item[]\refstepcounter{examplecntr}%
 {\bfseries Example~\theexamplecntr%
  \ifthenelse{\equal{#1}{}}{}{ (#1)}.
}}%
{\hfill$\Box$\end{trivlist}}

\newcounter{definitioncntr}
\newenvironment{definition}[1][]%
{\begin{trivlist}\item[]\refstepcounter{definitioncntr}%
{\bfseries Definition~\thedefinitioncntr%
 \ifthenelse{\equal{#1}{}}{}{ (#1)}.
}}%
{\hfill$\Box$\end{trivlist}}

\newcounter{theoremcntr}
\newenvironment{theorem}[1][]%
{\begin{trivlist}\item[]\refstepcounter{theoremcntr}%
{\bfseries Theorem~\thetheoremcntr%
  \ifthenelse{\equal{#1}{}}{}{ (#1)}.
}}%
{\hfill$\Box$\end{trivlist}}

\newcounter{lemmacntr}
\newenvironment{lemma}[1][]%
{\begin{trivlist}\item[]\refstepcounter{lemmacntr}%
{\bfseries Lemma~\thelemmacntr%
  \ifthenelse{\equal{#1}{}}{}{ (#1)}.
}}%
{\hfill$\Box$\end{trivlist}}

\newcounter{claimcntr}
{\begin{trivlist}\item[]\refstepcounter{claimcntr}%
{\bfseries Claim~\theclaimcntr%
  \ifthenelse{\equal{#1}{}}{}{ (#1)}.
}}%
{\hfill$\Box$\end{trivlist}}

\newcounter{corollarycntr}
\newenvironment{corollary}[1][]%
{\begin{trivlist}\item[]\refstepcounter{corollarycntr}%
{\bfseries Corollary~\thecorollarycntr%
  \ifthenelse{\equal{#1}{}}{}{ (#1)}.
}}%
{\hfill$\Box$\end{trivlist}}

\newcounter{propositioncntr}
\newenvironment{proposition}[1][]%
{\begin{trivlist}\item[]\refstepcounter{propositioncntr}%
{\bfseries Proposition~\thepropositioncntr%
  \ifthenelse{\equal{#1}{}}{}{ (#1)}.
}}%
{\hfill$\Box$\end{trivlist}}

\newenvironment{proof}{\begin{trivlist}\item[]{\bfseries Proof: }
 }{\hfill$\Box$\end{trivlist}}

\newenvironment{proofof}[1]{\begin{trivlist}\item[]{\bfseries Proof\ifthenelse{\equal{#1}{}}{}{ #1}:}
 }{\hfill$\Box$\end{trivlist}}

\newcommand{\eproofnegspace}{\\[-1.5\baselineskip]\rule{0em}{0ex}}

\newcounter{algorithmcntr}
\newenvironment{algorithm}[1][]%
{\normalsize\vspace{0.5ex}%
\refstepcounter{algorithmcntr}%
{\bfseries Algorithm~\thealgorithmcntr%
\ifthenelse{\equal{#1}{}}{:}{: #1}%
}}%
{}

\newcounter{subalgorithmcntr}[algorithmcntr]
\renewcommand{\thesubalgorithmcntr}{\thealgorithmcntr.\Alph{subalgorithmcntr}}
{\normalsize\vspace{0.5ex}%
\refstepcounter{subalgorithmcntr}%
{\bfseries Algorithm~\thesubalgorithmcntr%
\ifthenelse{\equal{#1}{}}{:}{: #1}%
}}%
{}

\newcounter{proglinecounter}
\newenvironment{pseudocode}%
    {\setcounter{proglinecounter}{0}%
     \begin{tabbing}123\=123\=123\=123\=123\=123\=123\=123\=123\=123\=123\= \kill}%
    {\end{tabbing}}
    {\setcounter{proglinecounter}{0}%
     \begin{tabbing}123\=123\=12\=12\=12\=12\=12\=12\=12\=12\=12\= \kill}%
    {\end{tabbing}}
\newcommand{\npcl}[1][]
    {\>\refstepcounter{proglinecounter}\arabic{proglinecounter}%
     \ifthenelse{\equal{#1}{}}{}{\label{#1}}\' \>}
\newcommand{\pkw}[1]{\textbf{#1}}    %

\newcounter{assertioncntr}
\newcommand{\assertlabel}[1]{\hfill\refstepcounter{assertioncntr}(A.\arabic{assertioncntr})\label{#1}}
\newcommand{\assertref}[1]{A.\ref{#1}}

\newcounter{extracntr}
\newcommand{\extralabel}[1]{\hfill\refstepcounter{extracntr}(E.\arabic{extracntr})\label{#1}}
\newcommand{\extraref}[1]{E.\ref{#1}}

\begin{document}

\title{Simultaneous Partial Inverses and Decoding\\ Interleaved Reed--Solomon Codes}

\author{%
Jiun-Hung~Yu,~\IEEEmembership{Member,~IEEE,} and Hans-Andrea Loeliger,~\IEEEmembership{Fellow,~IEEE}%
\thanks{%
Jiun-Hung Yu is with 
the Dept.\ of Electrical and Computer Engineering, 
National Chiao Tung University, Hsinchu, Taiwan 300, ROC.
H.-A.\ Loeliger is with
the Dept.\ of Information Technology and Electrical Engineering,
ETH Zurich, 8092 Z\"urich, Switzerland.

This paper was presented in part 
at the 2014 Allerton Conf.\ on Communication, Control, and Computing \cite{YuLoeliger2014} 
and at the 2015 IEEE Int.\ Symposium on Information Theory (ISIT) \cite{YuLoeliger2015isit}.

Copyright (c) 2017 IEEE. Personal use of this material is permitted.  
However, permission to use this material for any other purposes 
must be obtained from the IEEE by sending a request to pubs-permissions@ieee.org.
}}

\maketitle

\begin{abstract}
The paper introduces the simultaneous partial-inverse problem (SPI) for polynomials
and develops its application to decoding interleaved Reed--Solomon codes 
beyond half the minimum distance.
While closely related both to standard key equations and to well-known Pad\'e
approximation problems, the SPI problem stands out in several respects. 
First, the SPI problem has a unique solution (up to a scale factor),
which satisfies a natural degree bound. 
Second, the SPI problem can be transformed (monomialized) 
into an equivalent SPI problem where all moduli are monomials. 
Third, the SPI problem can be solved by an efficient algorithm 
of the Berlekamp--Massey type.
Fourth, decoding interleaved Reed--Solomon codes 
(or subfield-evaluation codes) beyond half the minimum distance 
can be analyzed in terms of a partial-inverse condition for the error pattern:
if that condition is satisfied, then the (true) error locator polynomial 
is the unique solution of a standard key equation 
and can be computed in many different ways,
including the well-known multi-sequence Berlekamp--Massey algorithm 
and
the SPI algorithm of this paper.
Two of the best performance bounds from the literature 
(the Schmidt--Sidorenko--Bossert bound and the Roth--Vontobel bound)
are generalized to hold for the partial-inverse condition
and thus to apply to several different decoding algorithms.
\end{abstract}


\section{Introduction}
\label{sec:SPIIntroduction}

This paper revolves around the following problem and develops
its application to decoding interleaved Reed--Solomon codes 
beyond half the minimum distance.

\begin{trivlist}\item{}{\bfseries Simultaneous Partial-Inverse (SPI) Problem:}
For $i=1,\ldots,L$,
let $b^{(i)}(x)$ and $m^{(i)}(x)$ be polynomials over some field $F$ 
with $\deg m^{(i)}(x) \geq 1$ and $\deg b^{(i)}(x)<\deg m^{(i)}(x)$. 
For fixed $\tau^{(i)}\in \Z$ with $0\leq \tau^{(i)} \leq \deg m^{(i)}(x)$, 
find a nonzero polynomial
$\Lambda(x)\in F[x]$ of the smallest degree such that
\begin{equation} \label{eqn:SPI}
\deg\!\Big(b^{(i)}(x)\Lambda(x) \bmod m^{(i)}(x)\Big)< \tau^{(i)}
\end{equation} 
for all $i\in\{1,\dots,L\}$.
\hfill$\Box$
\end{trivlist}
We will see that this problem has always a unique solution 
(up to a scale factor) 
and the solution satisfies 
\begin{equation}\label{Introduction:degreebound}
\deg \Lambda(x)\leq \sum_{i=1}^L \Big(\deg m^{(i)}(x)-\tau^{(i)}\Big).
\end{equation}

Moreover, we will see that the SPI problem for general moduli $m^{(i)}(x)$
can efficiently be transformed 
(``monomialized'') into an equivalent SPI problem 
with monomial moduli $m^{(i)}(x) = x^{\nu_i}$.

The special case%
\footnote{Except that \mbox{$b(x)=0$} was excluded in \cite{YuLoeliger2016IT}.}
$L=1$ was extensively discussed in \cite{YuLoeliger2016IT}.
In this paper, we address the generalization from $L=1$ to $L>1$, which is not obvious.

For $L>1$, the SPI problem appears to be new,
but it is closely related to a number of well-researched 
problems in coding and computer science including 
``key equations'' for interleaved Reed--Solomon codes 
\cite{SchSidoBossert2009,Nielsen2013},
the multi-sequence linear-feedback shift-register (MLFSR) problem 
of \cite{FengTzeng1989,FengTzeng1991,SchmidtSidorenko2006},
and generalizations of Pad\'e approximation problems 
\cite{GathenGerhard,RoSto:acmISSAC2016,PuRo:ipdISIT2017}.
However, none of these related problems shares all the mentioned 
properties (unique solution, degree bound, monomialization) of the SPI problem.

By developing the decoding of interleaved Reed--Solomon codes
around the SPI problem, we generalize and harmonize
a number of key ideas from the literature, as will be detailed below.

We will consider codes as follows. 
Let $F=F_q$ be a finite field with $q$ elements.
The codewords are $L\times n$ arrays over $F$ 
such that every row is a codeword in
some Reed--Solomon code over $F$.
We will only consider column errors, 
and we will not distinguish between columns with a single error
and columns with many errors.
The Reed--Solomon codes (for each row) consist of the codewords 
\begin{equation} \label{eqn:TheCode}
  \left\{ \big( a(\beta_0),\ldots,a(\beta_{n-1}) \big): a(x)\in F[x] \text{~with~} \deg a(x)<k \right\},
\end{equation}
where $\beta_0,\ldots,\beta_{n-1}$ are $n$ different elements of $F$.
Note that punctured Reed--Solomon codes are included 
and $\beta_\ell=0$ 
(for a single index $\ell$)
is allowed.
The dimension $k$ will be allowed to depend on the row.
However, for the further discussion in this section, we will assume 
that all row codes have the same dimension~$k$.

Such interleaved Reed--Solomon codes
can equivalently be viewed as punctured Reed--Solomon codes 
over $F_{q^L}$ 
simply by replacing $F[x]=F_q[x]$ in (\ref{eqn:TheCode}) by $F_{q^L}[x]$ 
while the evaluation points $\beta_0,\ldots,\beta_{n-1}$ remain in $F_q$ 
\cite{BrownMinderShokrollahi2004,ParvareshVardy2006isit,SchSidoBossert2009}.
Note that  symbol errors in $F_{q^L}$ correspond to  column errors in the array code.

Decoding such array codes (or subfield-evaluation codes)
beyond the Guruswami-Sudan decoding radius \cite{GuruswamiSudan} 
was studied in 
\cite{BleichenbacherKiayiasYung,ParvareshVardy2006isit,BrownMinderShokrollahi2004,ParvareshVardy,%
Metzner, Haslach19992001,SchSidoBossert2009, KurzweilSeidlHuber2011,RothVontobel2014}.
Following \cite{GuruswamiSudan}, 
some of these papers use list-decoding algorithms 
\cite{ParvareshVardy2006isit,ParvareshVardy} 
while others use unique-decoding algorithms that return at most one codeword
\cite{BrownMinderShokrollahi2004,BleichenbacherKiayiasYung,Metzner,Haslach19992001,SchSidoBossert2009,KurzweilSeidlHuber2011,RothVontobel2014}.
The best unique-decoding algorithms
can now correct $t$ errors (column errors or $F_{q^L}$-symbol errors) up to 
\begin{equation} \label{IRS:potentialerrorcorrectingbound}
t \leq  \frac{L}{L+1}\big(n-k\big)
\end{equation}
with high probability if $q$ is large 
\cite{BrownMinderShokrollahi2004,BleichenbacherKiayiasYung,SchSidoBossert2009}. 
For \mbox{$L\geq n-k-1$}, 
the bound (\ref{IRS:potentialerrorcorrectingbound}) becomes
\begin{equation} \label{eqn:tltnmk}
t < n-k,
\end{equation}
which cannot be improved. 
(For small $L$, however, improvements over (\ref{IRS:potentialerrorcorrectingbound}) 
have been demonstrated, cf.\ \cite{PuRo:ipdISIT2017} and the references therein.)

Specifically, for $t$ errors with random error values 
(uniformly distributed over all nonzero columns), 
the best bound on the probability $P_f$ of a decoding failure 
(due to Schmidt et al.\ \cite{SchSidoBossert2009}) is
\begin{equation} \label{eqn:SchSidoBossert2009bound2}
P_f  \leq \gamma \frac{q^{-L(n-k)+(L+1)t}}{q-1}
\end{equation}
with
\begin{equation} \label{eqn:SchSidoBossert2009boundPrefactor}
\gamma \eqdef \left(\frac{q^L-q^{-1}}{q^L-1}\right)^{t}.
\end{equation}
(Note that $\gamma>1$, but $\gamma\approx 1$ for any $t$ of interest.)
The bound (\ref{eqn:SchSidoBossert2009bound2}) 
implies that the decoding algorithm of \cite{SchSidoBossert2009} 
decodes up to (\ref{IRS:potentialerrorcorrectingbound}) errors 
with high probability if $q$ is large.

Another type of bound, not relying on randomness, 
uses the rank of the error matrix $E\in F^{L\times n}$ 
that corrupts the transmitted (array-) codeword 
\cite{RothVontobel2014}.
The decoding algorithm by Roth and Vontobel \cite{RothVontobel2014} 
corrects any $t$ (column) errors provided that 
\begin{equation} \label{IRS:guaranteedcorrectingbound}
t\leq \frac{n-k+\rank(E)-1}{2} \,,
\end{equation}
which beats the guarantee in \cite{SchSidoBossert2009} by a margin of $\rank(E)/2$.

Note that \cite{SchSidoBossert2009} and \cite{RothVontobel2014}
use different decoding algorithms, and the decoding algorithm of \cite{SchSidoBossert2009}
assumes cyclic Reed--Solomon codes (as row codes) where $m(x)=x^n-1$.

The bound (\ref{IRS:guaranteedcorrectingbound}) can also be used 
with random error values. For $t\leq L$, $\rank(E)$ is then likely to equal $t$,
in which case (\ref{IRS:guaranteedcorrectingbound}) reduces to (\ref{eqn:tltnmk});
for $t=n-k-1 \leq L$, 
(\ref{IRS:guaranteedcorrectingbound}) 
(by (\ref{eqn:FullRankProb})) yields the same bound as 
(\ref{eqn:SchSidoBossert2009bound2}) with $\gamma=1$,
which agrees with the bounds in
\cite{Metzner,KurzweilSeidlHuber2011}, 
where different decoding algorithms are used.

In this paper, 
we define a partial-inverse condition (Definition~\ref{def:PIC})
for the error pattern, 
which is always satisfied up to half the minimum distance and 
almost always 
satisfied almost up to the full minimum distance. 
If that condition is satisfied, then the (true) error locator polynomial
is the unique solution of a standard key equation 
and can thus be computed in several different ways.

Specifically,
we will show that (\ref{IRS:guaranteedcorrectingbound}) guarantees 
the partial-inverse condition to be satisfied.
For random error values (as above), 
the probability for this condition not to hold will be shown 
to be bounded by (\ref{eqn:SchSidoBossert2009bound2}), 
with the minor improvement 
that (\ref{eqn:SchSidoBossert2009boundPrefactor}) is replaced by $\gamma=1$.
In this way, the scope of both (\ref{eqn:SchSidoBossert2009bound2}) 
and (\ref{IRS:guaranteedcorrectingbound}) is widened.

The primary decoding algorithms for interleaved Reed--Solomon codes
are based on the MLFSR algorithm
by Feng and Tzeng \cite{FengTzeng1991} with corrections 
by Schmidt and Sidorenko \cite{SchmidtSidorenko2006,SchSidoBossert2009}
(but see also \cite{WangISIT2008}).
The complexity of this algorithm is $O(L(n-k)^2)$ 
additions and/or multiplications in $F$.
(Asymptotically faster algorithms have been proposed 
\cite{SiBo:fskew2014}
and will be discussed below.)
However, the MLFSR algorithm is restricted to monomial moduli,
which arise naturally from cyclic Reed--Solomon codes.

Beyond cyclic codes, 
for $L=1$, it is a classical result that decoding general Reed--Solomon 
codes can be reduced to a key equation with monomial modulus \cite{Roth},
which is amenable to the MLFSR algorithm.
(However, standard accounts of that method do not
allow an evaluation point $\beta_\ell$ to equal zero.)
For $L>1$, such a transformation was used in \cite{RothVontobel2014}.
In this paper, the same effect (without any constraints) 
is achieved by the monomialization of SPI problems, 
with the additional benefit that the partial-inverse condition 
is preserved.
This transformation can be carried out, either by the Euclidean algorithm 
or by the partial-inverse algorithm of \cite{YuLoeliger2016IT},
with complexity $O(L(n-k)^2)$.

Finally, we propose algorithms to solve the SPI problem. 
The basic SPI algorithm is of the Berlekamp--Massey type.
In the special case where $m^{(i)}(x)=x^{\nu_i}$,
it looks very much like, and has the same complexity as, the MLFSR algorithm 
\cite{FengTzeng1991,SchmidtSidorenko2006}.
However,
the two algorithms are different:
for \mbox{$L=1$}, the MLFSR algorithm of \cite{FengTzeng1991} and \cite{SchmidtSidorenko2006} 
reduces to the Berlekamp--Massey algorithm \cite{Massey}
while the proposed SPI algorithm (Algorithm~\ref{alg:BasicSPIAlg} of this paper) reduces 
to the reverse Berlekamp--Massey algorithm of \cite{YuLoeliger2016IT}. 

As shown in \cite{YuLoeliger2016IT}, the reverse Berlekamp--Massey algorithm
is easily translated into two other algorithms, one of them being a variation of the 
Euclidean algorithm by Sugiyama et al.\ \cite{Sugiyama}. 
The (reverse) Berlekamp--Massey algorithm and the Euclidean algorithm
may thus be viewed as two versions of the same algorithm.
In this paper, 
we extend this to $L>1$: 
by easy translations of Algorithm~\ref{alg:BasicSPIAlg}, 
we obtain two other algorithms (Algorithms \ref{alg:QSSPIAlg} and \ref{alg:RSSPIAlg}), 
one of which is of the Euclidean type and reminiscent 
of \cite{FengTzeng1989} rather than of \cite{FengTzeng1991}. 
(Yet another, quite different, ``Euclidean'' algorithm was proposed in \cite{BoBe:diEuclid2008c}.)
For $L>1$, no such connection between the (different) approaches of 
\cite{FengTzeng1989} and \cite{FengTzeng1991} has been described in the literature.
However, 
the (reverse) Berlekamp--Massey version for 
monomial (or monomialized) SPI problems stands out by having 
the lowest complexity.

As mentioned, 
asymptotically faster algorithms for the MLFSR problem 
have been presented in \cite{SiBo:fskew2014}
for cyclic row codes
and in \cite{Nielsen2013} 
for general row codes,
both achieving 
$O(L^3 (n-k) \log^2 (n-k) \log\log (n-k))$.
(Note that the asymptotic speed-up in $n-k$ is bought with the factor $L^3$.)
It seems likely that such asymptotically fast algorithms can also be developed 
for the SPI problem, but this is not addressed in the present paper. 
In any case, the algorithms from \cite{SiBo:fskew2014} and \cite{Nielsen2013}
also profit from the performance bounds and the monomialization scheme of this paper.

In summary, we demonstrate that the SPI problem 
allows to harmonize and to generalize 
a number of ideas 
from the literature on interleaved Reed--Solomon codes.
Specifically:
\begin{enumerate}
\item
A general SPI problem can be transformed into an equivalent SPI problem with $m^{(i)}(x) = x^{\nu_i}$.

When applied to decoding, this monomialization 
preserves the partial-inverse condition (with the associated guarantees).
We also show how the error evaluator polynomial 
(which is used, e.g., in Forney's formula, cf.\ Section~\ref{sec:ErrorLocator}) 
can be transformed accordingly.
\item
We show that the SPI problem can be solved by an efficient algorithm 
of the Berlekamp--Massey type. 

In the Appendix, 
we also show that this algorithm is easily translated into two other algorithms, 
one of which is of the Euclidean type. 
For the MLFSR problem with $L>1$, 
no such connection between the Berlekamp--Massey approach \cite{FengTzeng1991}
and the Euclidean approach \cite{FengTzeng1989}
has been demonstrated. 
However, for \mbox{$L>1$}, the (reverse) Berlekamp--Massey version 
has lower complexity than the other versions.
\item
Using the partial-inverse condition,
we prove 
the Schmidt--Sidorenko--Bossert bound (\ref{eqn:SchSidoBossert2009bound2}) (with $\gamma=1$)
and the Roth--Vontobel bound (\ref{IRS:guaranteedcorrectingbound})
for a range of algorithms including MLFSR algorithms
(such as, e.g., \cite{SchSidoBossert2009} and \cite{Nielsen2013})
and the SPI decoding algorithms of this paper.
\end{enumerate}

The paper is structured as follows. 
In Section~\ref{sec:SPIProblem}, 
we discuss the SPI problem without regard to any algorithms or applications.
In particular, we prove the degree bound (\ref{Introduction:degreebound}) 
and we discuss the monomialization of the SPI problem.
In Section~\ref{sec:SPIProblem4Decoding}, 
we consider the decoding of interleaved Reed--Solomon codes. 
The pivotal concept in this section is a partial-inverse condition 
for the error pattern, which guarantees that the (correct) error locator polynomial
can be computed by many different (well-known and new) algorithms. 
In Section~\ref{sec:DecodingFailureAnalysis}, 
the bounds (\ref{eqn:SchSidoBossert2009bound2}) and (\ref{IRS:guaranteedcorrectingbound}) 
are shown to apply to the partial-inverse condition.

In Section~\ref{sec:SPIAlg}, we return to the problem of actually 
solving SPI problems, for which 
we propose the reverse Berlekamp--Massey algorithm. 
In Section~\ref{sec:SPIDecoding}, we adapt and apply this algorithm 
to decoding interleaved Reed--Solomon codes.

The proof of the reverse Berlekamp--Massey algorithm 
is given in Appendix~\ref{section:proofAlgo}.
The other two versions (including the Euclidean version) 
of the algorithm are given in 
Appendix~\ref{section:QuotientRemainderSavingAlgs}.
Section~\ref{section:conclusion} concludes the paper.

The reader is assumed to be familiar with the basics of algebraic coding theory 
\cite{Blahut,JustesenHoholdt,Roth}
as well as with the notion of a ring homomorphism and its application to $F[x]$
\cite{Herstein}.

We will use the following notation.
The coefficient of $x^d$ of a polynomial $b(x)\in F[x]$ will be
denoted by $b_d$, and
the leading coefficient (i.e., the coefficient of
$x^{\deg b(x)}$) of a nonzero polynomial $b(x)$ will be denoted by $\lcf b(x)$.
We will use ``${\bmod}$'' both as in $r(x) = b(x) \bmod m(x)$ (the
remainder of a division) and as in $b(x) \equiv r(x)  \mod m(x)$
(a congruence modulo $m(x)$). 
We will also use ``$\odiv$'' for polynomial division: if
\begin{equation}
a(x) = q(x)m(x) + r(x)
\end{equation}
with $\deg r(x) < \deg m(x)$, then 
$q(x) = a(x) \odiv m(x)$ and $r(x) = a(x) \bmod m(x)$.

\section{About the Simultaneous Partial-Inverse Problem}
\label{sec:SPIProblem}

In this section, we consider the SPI problem 
without regard to any algorithms or applications.
The properties and facts that are proved here
are mostly straightforward generalizations from the case $L=1$ 
from \cite{YuLoeliger2016IT}, 
but Theorem~\ref{theorem:MonomializedSPI} (monomialization) 
requires some extra work.

\subsection{Basic Properties}
\label{sec:SPIBasics}

The SPI problem 
as defined in Section~\ref{sec:SPIIntroduction} has the following properties,
which will be heavily used throughout the paper.
\begin{proposition}\label{prop:existence}
The SPI problem has always a solution.
\end{proposition}

\begin{proof}
The polynomial $\Lambda(x)= m^{(1)}(x)\cdots m^{(L)}(x)$
satisfies $b^{(i)}(x)\Lambda(x) \bmod m^{(i)}(x) = 0$
for all $i$, which implies the existence of a solution for any $\tau^{(i)}\geq 0$.
\end{proof}

\begin{proposition}\label{propo:Uniqueness}
The solution $\Lambda(x)$ 
of an SPI problem 
is unique up to a scale factor in $F$.
\end{proposition}

\begin{proofof}{}
Let $\Lambda'(x)$ and $\Lambda''(x)$ be two solutions of the
problem, which implies $\deg \Lambda'(x)=\deg \Lambda''(x) \geq
0$. Define
\begin{IEEEeqnarray}{rCl}
r'^{(i)}(x)  & \eqdef &  b^{(i)}(x) \Lambda'(x) \bmod m^{(i)}(x) \label{eqn:ProofUniquenessr1}\IEEEeqnarraynumspace\\
r''^{(i)}(x)  & \eqdef &  b^{(i)}(x) \Lambda''(x) \bmod m^{(i)}(x)
\label{eqn:ProofUniquenessr2} \IEEEeqnarraynumspace
\end{IEEEeqnarray}
and consider
\begin{equation} \label{eqn:ProofUniquenessLambda3}
\Lambda(x)\eqdef \Big(\lcf \Lambda''(x) \Big) \Lambda'(x) -
\Big(\lcf \Lambda'(x) \Big) \Lambda''(x).
\end{equation}
Then
\begin{IEEEeqnarray}{rCl}
r^{(i)}(x)  & \eqdef &  b^{(i)}(x) \Lambda(x) \bmod m^{(i)}(x) \label{eqn:ProofUniquenessrx.0} \IEEEeqnarraynumspace\\
  & = & \Big(\lcf \Lambda''(x) \Big)  r'^{(i)}(x) - \Big(\lcf \Lambda'(x) \Big)  r''^{(i)}(x)
        \label{eqn:ProofUniquenessrx} \IEEEeqnarraynumspace
\end{IEEEeqnarray}
by the natural ring homomorphism $F[x] \rightarrow
F[x]/m^{(i)}(x)$. Clearly, (\ref{eqn:ProofUniquenessrx}) implies
that $\Lambda(x)$ also satisfies (\ref{eqn:SPI}) for every
\mbox{$1\leq i \leq L$}. But (\ref{eqn:ProofUniquenessLambda3})
implies $\deg \Lambda(x) < \deg\Lambda'(x)$, which is a
contradiction unless $\Lambda(x)=0$. Thus $\Lambda(x)=0$, which
means that $\Lambda'(x)$ equals $\Lambda''(x)$ up to a scale
factor.
\end{proofof}

\begin{proposition}[Degree Bound] \label{proposition:SolutionDegree}
If $\Lambda(x)$ solves the SPI problem,
then
\begin{equation} \label{eqn:MaxDegree}
\deg \Lambda(x) \leq \sum_{i=1}^{L} 
      \left(\deg m^{(i)}(x) - \tau^{(i)}\right).
\end{equation}
\eproofnegspace
\end{proposition}

\begin{proof}
The case $\tau^{(i)}=\deg m^{(i)}(x)$ for all $i$ is obvious. 
Otherwise, let $\nu_i \eqdef \deg m^{(i)}(x) - \tau^{(i)}$
and $\nu \eqdef \sum_{i=1}^L \nu_i$, and 
consider, for $i=1,\ldots,L$, the linear mappings
\begin{equation}
\varphi_i : F^{\nu+1} \rightarrow F^{\nu_i}
\end{equation}
given by
\begin{IEEEeqnarray}{rCl}
(\Lambda_0, \ldots, \Lambda_\nu)
& \mapsto & \Lambda(x) \eqdef \Lambda_0 + \Lambda_1 x + \ldots + \Lambda_\nu x^\nu 
            \IEEEeqnarraynumspace \\
& \mapsto & r^{(i)}(x) \eqdef b^{(i)}(x) \Lambda(x) \bmod m^{(i)}(x) \\
& \mapsto & (r^{(i)}_0, \ldots, r^{(i)}_{\deg m^{(i)}(x)-1}) \\
& \mapsto & (r^{(i)}_{\tau^{(i)}}, \ldots, r^{(i)}_{\deg m^{(i)}(x)-1}).
\end{IEEEeqnarray}
Note that a polynomial $\Lambda_0 + \Lambda_1 x + \ldots + \Lambda_\nu x^\nu$
satisfies (\ref{eqn:SPI}) if and only if 
$(\Lambda_0,\ldots,\Lambda_\nu) \in \ker \varphi_i$.
But 
\begin{IEEEeqnarray}{rCl}
\dim \left( \bigcap_{i=1}^L \ker \varphi_i \right) 
 & \geq & \nu+1 - \sum_{i=1}^L \nu_i \IEEEeqnarraynumspace\\
 & = & 1,
\end{IEEEeqnarray}
i.e., $\left( \bigcap_{i=1}^L \ker \varphi_i \right)$ is not trivial. 
There thus exists a nonzero polynomial
$\Lambda_0 + \Lambda_1 x + \ldots + \Lambda_\nu x^\nu$
that satisfies (\ref{eqn:SPI}) simultaneously for $i=1,\ldots,L$.
\end{proof}

For occasional later use, the right-hand side of (\ref{eqn:MaxDegree}) will be denoted by
\begin{equation}\label{eqn:DegreeBoundsymbol}
D\eqdef \sum_{i=1}^{L}\Big(\deg m^{(i)}(x)-\tau^{(i)}\Big).
\end{equation}

Another obvious bound on the degree of the solution is 
\begin{equation} \label{eqn:DegreeBoundLCM}
\deg \Lambda(x) \leq \deg\lcm\!\big( m^{(1)}(x), \ldots, m^{(L)}(x) \big),
\end{equation}
where $\lcm$ denotes the common multiple of the smallest degree.
In particular, we have
\begin{proposition}[Monomial-SPI Degree Bound]\hspace{-0.5em}\footnote{%
Proposition~\ref{proposition:MonomialDegreeBound} was pointed out by an anonymous reviewer.}
\label{proposition:MonomialDegreeBound}
Assume $m^{(i)}(x) = x^{\nu_i}$ for $i=1,\ldots,L$.
If $\Lambda(x)$ solves the SPI problem, then
\begin{equation} \label{eqn:MaxDegreeMonomial}
\deg \Lambda(x) \leq \max_{i=1,\ldots,L} \nu_i.
\end{equation}
\eproofnegspace
\end{proposition}
The right side of (\ref{eqn:MaxDegreeMonomial})
may be smaller or larger than (\ref{eqn:DegreeBoundsymbol}).

\subsection{Irrelevant Coefficients And Degree Reduction}

Let $\Lambda(x)$ be the solution of a given SPI problem
and let $u$ be a (nonnegative) integer such that 
\begin{equation}\label{eqn:BoundDegLambda}
u \geq \deg \Lambda(x). 
\end{equation}
Note that $u=D$ as in (\ref{eqn:DegreeBoundsymbol}) qualifies by (\ref{eqn:MaxDegree}).

\begin{proposition}[Irrelevant Coefficients] \label{proposition:irrelevantcoeff}
In the SPI problem, coefficients $b_\ell^{(i)}$ of
$b^{(i)}(x)$ with
\begin{equation} \label{eqn:IrrelevantCoeff_b}
\ell < \tau^{(i)} - u
\end{equation}
and coefficients $m_s^{(i)}$ of $m^{(i)}(x)$ with
\begin{equation} \label{eqn:IrrelevantCoeff_m}
s \leq  \tau^{(i)} - u
\end{equation}
have no effect on the solution $\Lambda(x)$.
\end{proposition}

\begin{proof}
From  (\ref{eqn:IrrelevantCoeff_b}) and (\ref{eqn:BoundDegLambda}), we obtain
\begin{equation}
\ell + \deg \Lambda(x) < \tau^{(i)},
\end{equation}
which proves the first claim.
As for the second claim, 
we begin by writing 
\begin{equation} \label{eqn:ProofIrrelCoeff2}
b^{(i)}(x) \Lambda(x) \bmod m^{(i)}(x) = b^{(i)}(x) \Lambda(x) - m^{(i)}(x) q^{(i)}(x)
\end{equation}
for some $q^{(i)}(x)\in F[x]$ with 
\begin{equation} \label{eqn:ProofIrrelCoeff3}
\deg q^{(i)}(x) < \deg \Lambda(x). 
\end{equation}
(If $q^{(i)}(x)\neq 0$, (\ref{eqn:ProofIrrelCoeff3}) follows from 
considering the leading coefficient of the right-hand side 
of (\ref{eqn:ProofIrrelCoeff2}) with $\deg b^{(i)}(x) < \deg m^{(i)}(x)$).
From (\ref{eqn:IrrelevantCoeff_m}), (\ref{eqn:ProofIrrelCoeff3}), and (\ref{eqn:BoundDegLambda}), 
we then obtain
\begin{equation} \label{eqn:ProofIrrelCoeff4}
s + \deg q^{(i)}(x) < \tau^{(i)}. 
\end{equation}
The second claim then follows 
from (\ref{eqn:ProofIrrelCoeff2}) and (\ref{eqn:ProofIrrelCoeff4}).
\end{proof}

Irrelevant coefficients according to Proposition~\ref{proposition:irrelevantcoeff}
may be set to zero without affecting the
solution $\Lambda(x)$. 
In fact, such coefficients can be stripped off as follows.
\begin{proposition}[Degree Reduction] \label{proposition:reducedPIproblem}
For any $u>0$ satisfying (\ref{eqn:BoundDegLambda}), 
let 
\begin{equation}
s^{(i)} \eqdef \max\{ \tau^{(i)} -u, 0 \}
\end{equation}
and define the polynomials $\tilde b^{(i)}(x)$ and $\tilde m^{(i)}(x)$ with
\begin{equation}
\tilde b_{\ell}^{(i)} \eqdef b_{\ell+s^{(i)}}^{(i)}
\end{equation}
and 
\begin{equation}
\tilde m_{\ell}^{(i)} \eqdef m_{\ell+s^{(i)}}^{(i)}
\end{equation}
for $\ell\geq 0$.
Then the modified SPI problem
with $b^{(i)}(x)$, $m^{(i)}(x)$, and $\tau^{(i)}$ replaced by
$\tilde b^{(i)}(x)$, $\tilde m^{(i)}(x)$, and $\tilde \tau^{(i)}\eqdef \tau^{(i)}-s^{(i)}$,
respectively,
has the same solution $\Lambda(x)$ 
as the original SPI problem.
In addition, we have
\begin{equation}\label{eqn:QuoUnChangedReducedEquation}
 b^{(i)}(x) \Lambda(x) \odiv m^{(i)}(x) = \tilde b^{(i)}(x) \Lambda(x) \odiv  \tilde  m^{(i)}(x) 
\end{equation}
\eproofnegspace
\end{proposition}

\begin{proof}
Consider an auxiliary simultaneous partial-inverse problem with
$b^{(i)}(x)$ replaced by $x^{s^{(i)}} \tilde b^{(i)}(x)$ and
$m^{(i)}(x)$ replaced by $x^{s^{(i)}} \tilde m^{(i)}(x)$ (and $\tau^{(i)}$ unchanged).
This auxiliary problem has the same solution
as the original problem by Proposition~\ref{proposition:irrelevantcoeff}.
The equivalence of this auxiliary problem with the 
modified problem is obvious from (\ref{eqn:ProofIrrelCoeff2}).
\end{proof}

\subsection{Monomialized SPI Problem}
\label{section:MonomializedSPI}

For a given SPI problem, let $u$ be 
a (nonnegative) integer 
that satisfies (\ref{eqn:BoundDegLambda}).
Moreover, let
$n^{(i)}\eqdef \deg m^{(i)}(x)$.

It turns out that the given SPI problem (with general moduli $m^{(i)}(x)$) 
can be transformed into another SPI problem where 
(\ref{eqn:SPI}) is replaced with
\begin{equation} \label{eqn:MonomializedSPI}
\deg\!\Big( \tilde b^{(i)}(x) \Lambda(x) \bmod x^{n^{(i)}-\tau^{(i)}+u} \Big)< u
\end{equation}
with $\tilde b^{(i)}(x)$ as defined below. 
The precise statement is given as Theorem~\ref{theorem:MonomializedSPI} below. 

We will need the additional condition
\begin{equation}\label{eqn:MonPIPCondd}
n^{(i)} - \tau^{(i)} + u > 0
\end{equation} 
for all $i\in \{1,\ldots,L\}$. 
Note that this condition does not entail any loss in generality: 
for $u>0$, (\ref{eqn:MonPIPCondd}) is always satisfied, 
and by (\ref{eqn:BoundDegLambda}), $u=0$ is admissible 
only if the SPI problem has the trivial solution $\Lambda(x)=1$.

The polynomial $\tilde b^{(i)}(x)$ in (\ref{eqn:MonomializedSPI}) is
defined as follows. 
Let 
\begin{IEEEeqnarray}{rCl}
\overline{b}^{(i)}(x) & \eqdef & x^{n^{(i)}-1}b^{(i)}(x^{-1}). \label{eqn:RevbPl}\\
\overline{m}^{(i)}(x) & \eqdef & x^{n^{(i)}}m^{(i)}(x^{-1}). \label{eqn:RevmPl}
\end{IEEEeqnarray}
Moreover, let $w^{(i)}(x)$ be the inverse of 
\begin{equation}\label{eqn:TrancatedBarmx}
\overline{m}^{(i)}(x)\bmod  x^{n^{(i)}-\tau^{(i)}+u}
\end{equation}
in the ring $F[x]/x^{n^{(i)}-\tau^{(i)}+u}$; this inverse exists because $\overline{m}^{(i)}(0)\neq 0$,
which implies that $\overline{m}^{(i)}(x)$ is relatively prime to  $x^{n^{(i)}-\tau^{(i)}+u}$.
Further, let
\begin{equation}\label{eqn:NewSymseq}
s^{(i)}(x)\eqdef \big(w^{(i)}(x) \overline{b}^{(i)}(x)\big) \bmod x^{n^{(i)}-\tau^{(i)}+u},
\end{equation} 
and finally 
\begin{equation}\label{proof:DefineTildebx}
\tilde b(x)\eqdef x^{n^{(i)}-\tau^{(i)}+u-1}s^{(i)}(x^{-1}).
\end{equation}

\begin{theorem}[Monomialized SPI Problem]\label{theorem:MonomializedSPI}
For a given SPI problem,
let $u$ be an integer satisfying both (\ref{eqn:BoundDegLambda}) and (\ref{eqn:MonPIPCondd}).
Then the modified SPI problem
where $b^{(i)}(x)$, $m^{(i)}(x)$, and $\tau^{(i)}$ are replaced by 
$\tilde b^{(i)}(x)$ (as defined above), $x^{n^{(i)}-\tau^{(i)}+u}$, and $u$, 
respectively, 
has the same solution $\Lambda(x)$ as the original SPI problem.
In addition, we have
\begin{equation}\label{eqn:QuoUnChangedMonoEquation}
 b^{(i)}(x) \Lambda(x) \odiv m^{(i)}(x) = \tilde b^{(i)}(x) \Lambda(x) \odiv  x^{n^{(i)}-\tau^{(i)}+u}
\end{equation}
\end{theorem}

Note that the computation of $\tilde b^{(i)}(x)$ requires the computation 
of $w^{(i)}(x)$ ($=$ the inverse of (\ref{eqn:TrancatedBarmx}) 
in $F[x]/x^{n^{(i)}-\tau^{(i)}+u}$),
which can be computed by the extended Euclidean algorithm or by 
the algorithms in \cite[Sec~IV]{YuLoeliger2016IT} 
(which coincide with the SPI algorithms of Section~\ref{sec:SPIAlg} for $L=1$).

\begin{proofof}{of Theorem~\ref{theorem:MonomializedSPI}}
Consider the original SPI problem (\ref{eqn:SPI}) and 
let $\Lambda(x)$ be its solution (which is unique up to a nonzero scale factor).
Let
\begin{equation}
r^{(i)}(x)\eqdef b^{(i)}(x) \Lambda(x) \bmod m^{(i)}(x),
\end{equation}
where $\deg r^{(i)}(x)<\tau^{(i)}$.
We then write
\begin{equation}
r^{(i)}(x) = b^{(i)}(x)\Lambda(x) - q^{(i)}(x)m^{(i)}(x) \label{eqn:RemainderPIVQuotient}
\end{equation}
for some (unique) $q^{(i)}(x)$ with 
\begin{equation} \label{eqn:MonomialProofDegqx}
\deg q^{(i)}(x)<\deg \Lambda(x)\leq u.
\end{equation}
Now let 
\begin{eqnarray} 
\overline{\Lambda}(x)& \eqdef & x^{u} \Lambda(x^{-1})  \label{eqn:RevLamPly}\\
\overline{q}^{(i)}(x) & \eqdef & x^{u-1}q^{(i)}(x^{-1})   \label{eqn:RevqPly} \\
\overline{r}^{(i)}(x) & \eqdef & x^{\tau^{(i)}-1}r^{(i)}(x^{-1}). \label{eqn:RevrPl}
\end{eqnarray}
By substituting $x^{-1}$ for $x$ in (\ref{eqn:RemainderPIVQuotient})
and multiplying both sides by $x^{n^{(i)}+u-1}$ (i.e., reversing (\ref{eqn:RemainderPIVQuotient})), we obtain
\begin{equation}
x^{n^{(i)}+u-\tau^{(i)}}\overline{r}^{(i)}(x)
=\overline{b}^{(i)}(x) \overline{\Lambda}(x) - \overline{q}^{(i)}(x)\overline{m}^{(i)}(x).
\end{equation}
We then have
\begin{equation}
\overline{b}^{(i)}(x) \overline{\Lambda}(x) \equiv \overline{q}^{(i)}(x)\overline{m}^{(i)}(x) \mod x^{n^{(i)}-\tau^{(i)}+u}
\end{equation}
and thus 
\begin{equation} \label{eqn:MonomialedEqn}
s^{(i)}(x) \overline{\Lambda}(x) \equiv \overline{q}^{(i)}(x) \mod x^{n^{(i)}-\tau^{(i)}+u}
\end{equation}
with $s^{(i)}(x)$ defined in (\ref{eqn:NewSymseq}).
Note that $\deg\overline\Lambda(x)\leq u$ 
and $\deg\overline q^{(i)}(x)<u$ from (\ref{eqn:MonomialProofDegqx}).

We now write (\ref{eqn:MonomialedEqn}) as
\begin{equation} \label{eqn:MonomialedEqn.1}
s^{(i)}(x) \overline{\Lambda}(x) = \overline p^{(i)}(x)x^{n^{(i)}-\tau^{(i)}+u} + \overline{q}^{(i)}(x)
\end{equation}
for some (unique) $\overline p^{(i)}(x)$ 
with $\deg \overline p^{(i)}(x)<\deg\overline{\Lambda}(x)\leq u$, and let
\begin{equation}
p^{(i)}(x)\eqdef x^{u-1}\overline p^{(i)}(x^{-1}).
\end{equation}
By substituting $x^{-1}$ for $x$ in (\ref{eqn:MonomialedEqn.1})
and multiplying both sides by $ x^{n^{(i)}-\tau^{(i)}+2u-1}$, we obtain
\begin{equation}\label{eqn:MonomialedEqnWanted}
\tilde{b}^{(i)}(x) \Lambda(x) = x^{n^{(i)}-\tau^{(i)}+u}{q}^{(i)}(x) + p^{(i)}(x),
\end{equation}
from which we have
\begin{equation}\label{eqn:MonomializedSPIProblem} 
\deg\!\Big(\tilde b^{(i)}(x) \Lambda(x) \bmod  x^{n^{(i)}-\tau^{(i)}+u}\Big)< u.
\end{equation} 
We have arrived at the modified SPI problem.

Now, let $\tilde \Lambda(x)$ denote the solution of the modified SPI problem,
which implies
\begin{equation}\label{eqn:TheTildeLambdaDegree}
\deg \tilde \Lambda(x)\leq \deg \Lambda(x).
\end{equation}
In the following, we will show
\begin{equation}\label{eqn:TheGoaltildeLambda}
\deg \tilde \Lambda(x)\geq \deg \Lambda(x).
\end{equation}
When this is established, we have $\deg \tilde \Lambda(x)=\deg \Lambda(x)$;
thus $\Lambda(x)$ solves the modified SPI problem,  
and (\ref{eqn:QuoUnChangedMonoEquation}) is obvious from (\ref{eqn:MonomialedEqnWanted}).

It remains to prove (\ref{eqn:TheGoaltildeLambda}). 
We begin by writing
\begin{equation}
\deg\!\Big(\tilde b^{(i)}(x) \tilde\Lambda(x) \bmod  x^{n^{(i)}-\tau^{(i)}+u}\Big)< u
\end{equation} 
with $\deg\tilde \Lambda(x)\leq u$ because of (\ref{eqn:TheTildeLambdaDegree}).
Now, let
\begin{equation}
 \tilde p^{(i)}(x)\eqdef \tilde b^{(i)}(x) \tilde\Lambda(x) \bmod  x^{n^{(i)}-\tau^{(i)}+u},
\end{equation}
where $\deg\tilde p^{(i)}(x)<u$. We then write
\begin{equation}\label{eqn:revbpq}
\tilde{b}^{(i)}(x) \tilde\Lambda(x) =  x^{n^{(i)}-\tau^{(i)}+u}\tilde{q}^{(i)}(x) + \tilde p^{(i)}(x)
\end{equation}
for some (unique) $\tilde{q}^{(i)}(x)$ with $\deg \tilde{q}^{(i)}(x)<\deg \tilde \Lambda(x)$.

Further, let 
\begin{eqnarray}
\hat{\Lambda}(x)& \eqdef & x^{u} \tilde\Lambda(x^{-1})  \label{eqn:RevtildeLamPly}\\
\hat{p}^{(i)}(x) & \eqdef & x^{u-1}\tilde p^{(i)}(x^{-1})   \label{eqn:RevtildepPly} \\
\hat{q}^{(i)}(x) & \eqdef & x^{u-1}\tilde q^{(i)}(x^{-1}). \label{eqn:RevtildeqPly}
\end{eqnarray}
By substituting $x^{-1}$ for $x$ in (\ref{eqn:revbpq})
and multiplying both sides by $x^{n^{(i)}-\tau^{(i)}+2u-1}$, we obtain
\begin{equation}
{s}^{(i)}(x) \hat\Lambda(x) =  x^{n^{(i)}-\tau^{(i)}+u}\hat{p}^{(i)}(x) + \hat q^{(i)}(x),
\end{equation}
where ${s}^{(i)}(x)$ is obtained from (\ref{proof:DefineTildebx}). 
It follows that 
\begin{equation}
s^{(i)}(x)  \hat\Lambda(x)  \equiv \hat{q}^{(i)}(x) \mod  x^{n^{(i)}-\tau^{(i)}+u}
\end{equation}
and therefore 
\begin{equation}\label{eqn:revwbhatLamhatq}
w^{(i)}(x) \overline{b}^{(i)}(x) \hat{\Lambda}(x) \equiv \hat{q}^{(i)}(x) \mod  x^{n^{(i)}-\tau^{(i)}+u}.
\end{equation}
By multiplying both sides of (\ref{eqn:revwbhatLamhatq}) by $\overline{m}^{(i)}(x)$, we obtain
\begin{equation}\overline{b}^{(i)}(x) \hat{\Lambda}(x) \equiv \hat{q}^{(i)}(x)\overline{m}^{(i)}(x) \mod  x^{n^{(i)}-\tau^{(i)}+u},
\end{equation}
and therefore 
\begin{equation}\label{eqn:EndrevbhatLamhatqhatr}
\overline{b}^{(i)}(x) \hat{\Lambda}(x)-\hat{q}^{(i)}(x)\overline{m}^{(i)}(x)= x^{n^{(i)}-\tau^{(i)}+u}\hat r^{(i)}(x) 
\end{equation}
holds for some (unique) $\hat r^{(i)}(x)$ with $\deg \hat r^{(i)}(x)<\tau^{(i)}$.
Let
\begin{equation} 
\tilde r^{(i)}(x)\eqdef  x^{\tau^{(i)}-1} \hat r^{(i)}(x^{-1}).
\end{equation}
By substituting $x^{-1}$ for $x$ in (\ref{eqn:EndrevbhatLamhatqhatr})
and multiplying both sides by $x^{n^{(i)}+u-1}$, we obtain
\begin{equation}
{b}^{(i)}(x) \tilde{\Lambda}(x)-\tilde{q}^{(i)}(x){m}^{(i)}(x)=\tilde r^{(i)}(x) 
\end{equation}
and therefore
\begin{equation}\label{eqn:TheEndSPIProb}
\deg\!\Big(b^{(i)}(x) \tilde\Lambda(x) \bmod m^{(i)}(x)\Big)<\tau^{(i)}.
\end{equation} 
Thus $\tilde\Lambda(x)$ satisfies (\ref{eqn:SPI}), 
and (\ref{eqn:TheGoaltildeLambda}) follows.
\end{proofof}

\section{Utilizing the SPI Problem for Decoding Interleaved Reed--Solomon Codes}
\label{sec:SPIProblem4Decoding}

\subsection{About Interleaved Reed--Solomon Codes}
\label{sec:IRSCodes}

We first establish some (more or less standard) concepts 
and the pertinent notation.

\subsubsection{Array Codes and Evaluation Isomorphism}

Let $F=F_q$ be a finite field with $q$ elements. 
We consider array codes as defined in Section~\ref{sec:SPIIntroduction} 
where codewords are $L\times n$ arrays over $F$ such that 
each row is a codeword in some Reed--Solomon code 
as in (\ref{eqn:TheCode}) 
with blocklength $n$ and dimension $k^{(i)}$, $i\in \{1,\ldots,L\}$.
Let 
\begin{equation} \label{eqn:Mx}
m(x) \eqdef \prod_{\ell=0}^{n-1}(x-\beta_\ell),
\end{equation} 
where $\deg m(x)=n$.
Let $\psi$ be the evaluation mapping
\begin{equation} \label{eqn:DefMappingPsi}
\psi: F[x]/m(x)\rightarrow F^n: a(x) \mapsto \big( a(\beta_0),\ldots, a(\beta_{n-1}) \big),
\end{equation}
which is a ring isomorphism. 
The row code (\ref{eqn:TheCode}) can then be described as
\begin{equation} \label{eqn:RowCodeWords}
\{ c\in F^n: \deg\psi^{-1}(c) < k^{(i)} \}.
\end{equation}

The standard definition of Reed--Solomon codes requires, in addition, that  
\begin{equation}\label{eqn:CondForDFT}
\beta_\ell = \alpha^\ell \text{~for~}\ell=0,\ldots,n-1, 
\end{equation}
where $\alpha\in F$ is a primitive $n$-th root of unity.
This additional condition implies 
\begin{equation}\label{eqn:mxofDFT}
m(x)=x^n-1, 
\end{equation}
which makes the code cyclic 
and turns $\psi$ into a discrete Fourier transform \cite{Blahut}. 
However, (\ref{eqn:CondForDFT})
and (\ref{eqn:mxofDFT}) will not be required below. In particular, the set 
$\{\beta_0,\ldots, \beta_{n-1}\}$ will be permitted to contain 0.

In general, the inverse mapping $\psi^{-1}$ may be computed 
by Lagrange interpolation or according to the Chinese remainder theorem.

We are primarily interested in the special case
where all row codes have the same dimension $k^{(i)} = k$. 
Nonetheless, we allow the general case, for which we define
\begin{equation}\label{eq:defineMaxk}
k_\text{max}\eqdef\max\{k^{(i)}, 1\leq i\leq L\},
\end{equation} 
\begin{equation}\label{eq:defineMink}
k_\text{min}\eqdef\min\{k^{(i)}, 1\leq i\leq L\},
\end{equation} 
and
\begin{equation}\label{eq:defineAvgk}
k_\text{avg}\eqdef \frac{1}{L}\sum_{i=1}^L k^{(i)}.
\end{equation}
Note that the overall rate of the array code is $k_\text{avg}/n$.
Throughout the paper, we will assume
\begin{equation} \label{eqn:kmaxsmallern}
k_\text{max} < n.
\end{equation}

\subsubsection{Notation for Individual Rows and Error Support}\label{sec:NotationForIndRows}

Let $Y = C+E\in F^{L\times n}$ be the received word where $C\in F^{L\times n}$ 
is the transmitted (array-) codeword and $E\in F^{L\times n}$ is the error pattern.
Further, let $y^{(i)}$ be the $i$-th row of the matrix $Y$, 
let $c^{(i)}$ be the $i$-th row of $C$, and
let $e^{(i)}$ be the $i$-th row of $E$. We then have
$y^{(i)}=c^{(i)}+e^{(i)}$,
$i=1,\ldots, L$, and therefore
\begin{equation} \label{eqn:Yix}
Y^{(i)}(x)=C^{(i)}(x)+E^{(i)}(x)
\end{equation}
where $Y^{(i)}(x)\eqdef \psi^{-1}(y^{(i)})$, $C^{(i)}(x) \eqdef
\psi^{-1}(c^{(i)}),$ and $E^{(i)}(x) \eqdef \psi^{-1}(e^{(i)})$.
Note that $\deg E^{(i)}(x) < \deg m(x)=n$ and $\deg C^{(i)}(x) < k^{(i)}$.

We will index the columns of codewords and error patterns 
beginning with zero as in 
$E = (e_0,\ldots,e_{n-1})$, 
and we define
\begin{equation} \label{eqn:ErrorSupportU}
U_E \eqdef \big\{ \ell \in \{ 0,\ldots, n-1\}: e_\ell \neq 0 \big\},
\end{equation}
the index set of the nonzero columns of $E$.
Note that $e_\ell = 0$ if and only if
\begin{equation} \label{eqn:ErrorSupportEval}
E^{(i)}(\beta_\ell) = 0 \text{ for all $i\in \{ 1,\ldots,L \}$.}
\end{equation}
We will only consider column errors, and we will not distinguish between 
columns with a single error and columns with many errors.

\subsubsection{Error Locator Polynomial and Interpolation}
\label{sec:ErrorLocator}

We define the error locator polynomial as%
\footnote{In the literature, the error locator polynomial is more often defined 
with $(x-\beta_j)$ in (\ref{eq:IRSerrorlocator})
replaced by $(1-x\beta_j)$.} 
\begin{equation} \label{eq:IRSerrorlocator}
\Lambda_E(x)\eqdef \prod_{\ell\in U_E}(x-\beta_\ell).
\end{equation}
(In particular, $\Lambda_E(x)=1$ if $E=0$.)
Note that
\begin{equation}\label{eqn:ErrorNumAndLamdDegree}
\deg \Lambda_E(x) = |U_E| = \text{number of column errors}.
\end{equation}

If $\Lambda_E(x)$ is known and satisfies $\deg \Lambda_E \leq n-k_\text{max}$, 
the polynomial $C^{(i)}(x)$ %
for each $i\in\{1,\ldots,L\}$ can be recovered in many different ways 
(cf.\ the discussion in \cite{YuLoeliger2016IT}), 
e.g., by means of
\begin{equation}\label{eqn:interpolation2011}
C^{(i)}(x)= \frac{Y^{(i)}(x) \Lambda_E(x) \bmod m(x)}{ \Lambda_E(x)}
\end{equation}
or by means of
\begin{equation} \label{eqn:IRSRecoverCxByDiv}
C^{(i)}(x) = Y^{(i)}(x) \bmod \tilde m(x)
\end{equation}
with $\tilde m(x)\eqdef m(x)/\Lambda_E(x)$
according to \cite[Proposition~9]{YuLoeliger2016IT}. 
For large $L$, (\ref{eqn:IRSRecoverCxByDiv}) may be more attractive.

If we need to recover the actual codeword $c^{(i)}$ (rather than just the polynomial $C^{(i)}(x)$), 
interpolation according to (\ref{eqn:interpolation2011}) or (\ref{eqn:IRSRecoverCxByDiv}) requires 
the additional computation of $\psi(C^{(i)}(x))$. 
Alternatively, it may be attractive to 
compute the error pattern $e^{(i)}=y^{(i)}-c^{(i)}$ 
from
\begin{equation}\label{eqn:TheMamyQi}
Q^{(i)}(x)\eqdef Y^{(i)}(x)\Lambda_E(x)\odiv m(x)
\end{equation}
and Forney's formula
\begin{equation}\label{eqn:ForneyFormula}
e_\ell^{(i)}
  \eqdef \left\{ \begin{array}{ll}
         0 & \text{if $\Lambda_E(\beta_\ell)\neq0$} \\
         \frac{Q^{(i)}(\beta_\ell)m'(\beta_\ell)}{\Lambda_E'(\beta_\ell)} & \text{if $\Lambda_E(\beta_\ell)=0$}
        \end{array}\right.
\end{equation}
for $\ell=0,1,\ldots, n-1$, where $\Lambda_E'(x)$ and $m'(x)$ 
denote the formal derivatives of $\Lambda_E(x)$ and $m(x)$, respectively.

\subsubsection{Shiozaki--Gao Error-Locator Equation \cite{Shiozaki,Gao}}
\label{subsection:SPIELEquation}

From (\ref{eqn:ErrorSupportEval}) and (\ref{eq:IRSerrorlocator}), 
we have
\begin{equation}\label{eqn:ClassicalKeyEqn}
E^{(i)}(x) \Lambda_E(x) \bmod m(x) = 0
\end{equation}
and thus
\begin{equation}
Y^{(i)}(x) \Lambda_E(x) \bmod m(x) = C^{(i)}(x) \Lambda_E(x) \bmod m(x).
\end{equation}
If $|U_E|\leq n-k_\text{max}$, we obtain
\begin{equation} \label{IRS:KeyDecdoingDegCond}
\deg \!\big( Y^{(i)}(x) \Lambda_E(x) \bmod m(x) \big) 
  < k^{(i)} + |U_E|
 \end{equation}
for all $i\in\{1,\ldots,L\}$.

\subsection{Partial-Inverse Conditions}
\label{sec:PIC}

We now begin to develop the specific approach of this paper.
We first note that any finite set of inequalities of the form (\ref{eqn:SPI}) 
(indexed by $i\in\{1,\ldots,L\}$, 
with suitable polynomials $b^{(i)}(x)$ and $m^{(i)}(x)$, 
and with suitable integers $\tau^{(i)}\in\Z$)
implicitly defines an SPI problem. 
\begin{definition}[SPI Solution]
The solution of such an SPI problem will be called 
the \emph{SPI solution} of the inequalities. 
\end{definition}
Clearly, the SPI solution is unique, up to a nonzero scale factor in $F$.

For a given code as above and some arbitrary, but fixed, 
error pattern $E\in F^{L \times n}$, 
we now consider the conditions
\begin{equation} \label{eqn:PICond}
\deg \!\Big( E^{(i)}(x) \Lambda(x) \bmod m(x) \Big) < k^{(i)} + |U_E|,
\end{equation} 
$i = 1,\ldots, L$.
Note that
$\Lambda(x)=\Lambda_E(x)$ satisfies (\ref{eqn:PICond}) by (\ref{eqn:ClassicalKeyEqn}). 

\begin{definition}[Partial-Inverse Condition] \label{def:PIC}
$E$ satisfies the \emph{partial-inverse condition} 
if $|U_E|\leq n-k_\text{max}$ and
the SPI solution of (\ref{eqn:PICond}) is $\Lambda_E(x)$.
\end{definition}

We will see in Section~\ref{sec:DecodingFailureAnalysis}
that $E$ always satisfies the partial-inverse condition 
if $|U_E| \leq (n-k_\text{max})/2$,
and $E$ is very likely to satisfy the partial-inverse condition 
if $|U_E| \leq \frac{L}{L+1}(n - k_\text{avg})$.

We now proceed to derive several equivalent formulations of the partial-inverse condition: 
the received-word version (Proposition~\ref{proposition:PICond}), 
the syndrome version (Theorem~\ref{theorem:ReducedPICond}), 
and the monomialized version (Theorem~\ref{theorem:MonomializedPICondQu}).

\begin{proposition}[Received-Word Partial-Inverse Condition]\label{proposition:PICond}%
Let $Y=C+E \in F^{L\times n}$ where $C$ is a codeword and 
$|U_E|\leq n-k_\text{max}$. Then 
$E$ satisfies the partial-inverse condition
if and only if the SPI solution of 
\begin{equation} \label{eqn:PICondY}
\deg \!\Big( Y^{(i)}(x) \Lambda(x) \bmod m(x) \Big) < k^{(i)} + |U_E|,
\end{equation}
$i = 1,\ldots,L$,
is $\Lambda_E(x)$.
\end{proposition}

\begin{proof}
Consider any fixed $E$ with $|U_E|\leq n-k_\text{max}$. 
From (\ref{eqn:ClassicalKeyEqn}), 
the solution of the SPI problem (\ref{eqn:PICond}) has degree at most $|U_E|$. 
From (\ref{IRS:KeyDecdoingDegCond}), 
the solution of the SPI problem (\ref{eqn:PICondY}) has degree at most $|U_E|$ as well.
But for any $\Lambda(x)$ with $\deg \Lambda(x) \leq |U_E|$, we have
\begin{equation}
\deg \!\Big( C^{(i)}(x) \Lambda(x) \bmod m(x) \Big) < k^{(i)} + |U_E|
\end{equation}
for $i=1,\ldots,L$. 
Thus $\Lambda(x)$ satisfies (\ref{eqn:PICond}) if and only if it satisfies (\ref{eqn:PICondY}).
\end{proof}

The conditions (\ref{eqn:PICondY}) can be simplified as follows.
For a given code and received word $Y=C+E$, let 
\begin{IEEEeqnarray}{rCl}
S^{(i)}(x) 
& \eqdef & Y^{(i)}_{k^{(i)}} + Y^{(i)}_{{k^{(i)}}+1}x + \ldots + Y^{(i)}_{n-1} x^{n-{k^{(i)}}-1}
           \IEEEeqnarraynumspace\label{eqn:Truncatebx}\\
& = & E^{(i)}_{k^{(i)}} + E^{(i)}_{{k^{(i)}}+1}x + \ldots + E^{(i)}_{n-1} x^{n-{k^{(i)}}-1}
\end{IEEEeqnarray}
(the syndrome polynomials) and
\begin{equation} \label{eqn:Truncatemx}
\tilde m^{(i)}(x)  \eqdef m_{k^{(i)}} + m_{k^{(i)}+1}x + \ldots + m_n x^{n-k^{(i)}}. 
\end{equation}

\begin{theorem}[Syndrome-Based Partial-Inverse Condition]\label{theorem:ReducedPICond}
Let $Y=C+E \in F^{L\times n}$ where $C$ is a codeword and
$|U_E|\leq n-k_\text{max}$. Then
$E$ satisfies the partial-inverse condition 
if and only if the SPI solution of 
\begin{equation} \label{eqn:ReducedPICond}
\deg \!\Big( S^{(i)}(x) \Lambda(x) \bmod \tilde m^{(i)}(x) \Big) < |U_E|,
\end{equation} 
$i = 1,\ldots,L$,
is $\Lambda_E(x)$.
Moreover, if $E$ satisfies the partial-inverse condition, then
\begin{equation}\label{eqn:ReducedPICondQu}
 Y^{(i)}(x) \Lambda_E(x) \odiv m(x) = S^{(i)}(x) \Lambda_E(x) \odiv  \tilde  m^{(i)}(x).
\end{equation}
\end{theorem}
Note that (\ref{eqn:ReducedPICondQu}) is the polynomial $Q^{(i)}(x)$ in (\ref{eqn:TheMamyQi}), 
which is required in Forney's formula (\ref{eqn:ForneyFormula}).

For $|U_E|>0$, the proof of Theorem~\ref{theorem:ReducedPICond}
follows from Proposition~\ref{proposition:reducedPIproblem}
with $u = |U_E|$ and $s^{(i)}=k^{(i)}$.
For $|U_E|=0$, (\ref{eqn:PICond}), (\ref{eqn:ReducedPICond}),
and (\ref{eqn:ReducedPICondQu})
are always satisfied.

The partial-inverse condition for general $m(x)$ can be reduced to
a partial-inverse condition for $m(x)=x^{n-{k^{(i)}}}$ as follows.
For a given code and received word $Y=C+E$, let 
\begin{equation} \label{eqn:MonomPICond}
\deg \!\Big( \breve S^{(i)}(x) \Lambda(x) \bmod x^{n-k^{(i)}} \Big) < |U_E|
\end{equation}
be the monomialization of (\ref{eqn:ReducedPICond}) 
according to Theorem~\ref{theorem:MonomializedSPI}
(with $b^{(i)}(x)=S^{(i)}(x)$, $\tilde b^{(i)}(x) = \breve S^{(i)}(x)$, 
$u=|U_E|$, $n^{(i)}=n - k^{(i)}$ and $n^{(i)} - \tau^{(i)} + u = n - k^{(i)}$).

Note that the computation of $\breve S^{(i)}(x)$ from $S^{(i)}(x)$ 
does not depend on $|U_E|$. 
Note also that the condition (\ref{eqn:MonPIPCondd}) translates to (\ref{eqn:kmaxsmallern}).

\begin{theorem}[Monomialized Partial-Inverse Condition] \label{theorem:MonomializedPICondQu}
Let $Y=C+E \in F^{L\times n}$ where $C$ is a codeword and 
$|U_E|\leq n-k_\text{max}$. Then 
$E$ satisfies the partial-inverse condition 
if and only if the SPI solution of (\ref{eqn:MonomPICond}) is $\Lambda_E(x)$.
Moreover, if $E$ satisfies the partial-inverse condition, then
\begin{equation}\label{eqn:MonomializedPICondQu}
 Y^{(i)}(x) \Lambda_E(x) \odiv m(x) = \breve S^{(i)}(x) \Lambda_E(x) \odiv  x^{n-k^{(i)}}.
\end{equation}
\end{theorem}
Note that (\ref{eqn:MonomializedPICondQu}) is the polynomial $Q^{(i)}(x)$ 
in (\ref{eqn:TheMamyQi}) and (\ref{eqn:ForneyFormula}).
The proof of Theorem~\ref{theorem:MonomializedPICondQu} 
is immediate from  Theorems \ref{theorem:MonomializedSPI}
and~\ref{theorem:ReducedPICond}. 

The complexity of computing $\breve S^{(i)}(x)$
is determined by the complexity of computing the inverse (\ref{eqn:TrancatedBarmx}) 
and the multiplication (\ref{eqn:NewSymseq}), both mod $x^{n-k^{(i)}}$. 
If the inverse is computed with the Euclidean algorithm 
(or with the partial-inverse algorithm of \cite{YuLoeliger2016IT}), 
the complexity of these computations is 
$O\big( L(n-k^{(i)})^2 \big)$ additions and/or multiplications in $F$.
(Asymptotical speed-ups are certainly possible, but outside the scope of this paper.)

In summary, for $|U_E|\leq n-k_\text{max}$, 
the SPI solutions of (\ref{eqn:PICond}), (\ref{eqn:PICondY}), 
(\ref{eqn:ReducedPICond}), and (\ref{eqn:MonomPICond}) coincide. 
In Section~\ref{sec:DecodingFailureAnalysis}, we will see that 
this SPI solution is very likely to be $\Lambda_E(x)$.

\subsection{Computing the Error Locator Polynomial}
\label{sec:ComputingEL}

If $E$ satisfies the partial-inverse condition,
$\Lambda_E(x)$ can be computed in many different ways. 
Let $S^{(i)}(x)$, $\tilde m^{(i)}(x)$, and $\breve S^{(i)}(x)$ be defined 
as in (\ref{eqn:Truncatebx}), (\ref{eqn:Truncatemx}), and (\ref{eqn:MonomPICond}), respectively.

\begin{proposition}[Key Equations] \label{proposition:KeyEqs}
If $E$ satisfies the partial-inverse condition,
then $\Lambda_E(x)$ is the unique (up to a scale factor) 
nonzero polynomial $\Lambda(x)$ of smallest degree such that
\begin{equation} \label{eqn:ReducedKeyEqn}
\deg \!\Big( S^{(i)}(x) \Lambda(x) \bmod \tilde m^{(i)}(x) \Big) < \deg\Lambda(x)
\end{equation}
or, equivalently, such that
\begin{equation} \label{eqn:MonKeyEqn}
\deg \!\Big( \breve S^{(i)}(x) \Lambda(x) \bmod x^{n-k^{(i)}} \Big) < \deg\Lambda(x)
\end{equation}
for all $i \in \{ 1,\ldots,L \}$.
\end{proposition}

\begin{proof}
Assume that 
$E$ satisfies the partial-inverse condition.
Then $\Lambda(x) = \Lambda_E(x)$ is the SPI solution of
(\ref{eqn:ReducedPICond}) and (\ref{eqn:MonomPICond}).
\end{proof}

Note that (\ref{eqn:MonKeyEqn}) is a standard key equation, 
which can be derived and solved in many different ways. 
The point of Proposition~\ref{proposition:KeyEqs}
is the guarantee from the partial-inverse condition.

In order to compute $\Lambda_E(x)$ with an MLFSR algorithm 
as in \cite{FengTzeng1989,FengTzeng1991,SchmidtSidorenko2006},
we need a slightly different formulation.

\begin{proposition}[MLFSR Problems] \label{proposition:MLFSREqns}
Assume that 
$E$ satisfies the partial-inverse condition.
Then the smallest integer $\kappa$ such that there exists 
a nonzero polynomial $\Lambda(x)$ with $\deg \Lambda(x)\leq \kappa$
such that
\begin{equation} \label{eqn:ReducedMLFSREqn}
\deg \!\Big( S^{(i)}(x) \Lambda(x) \bmod \tilde m^{(i)}(x) \Big) < \kappa
\end{equation}
or, equivalently, such that
\begin{equation} \label{eqn:MonMLFSREqn}
\deg \!\Big( \breve S^{(i)}(x) \Lambda(x) \bmod x^{n-k^{(i)}} \Big) < \kappa
\end{equation}
for all $i \in \{ 1,\ldots,L \}$
is $\kappa=\deg \Lambda_E(x)$.
Moreover, 
$\Lambda(x) = \Lambda_E(x)$ is the unique (up to a scale factor) such polynomial 
$\Lambda(x)$.
\end{proposition}

\begin{proof}
Again, the proof follows from noting that 
$\Lambda(x) = \Lambda_E(x)$ is the SPI solution of
(\ref{eqn:ReducedPICond}) and (\ref{eqn:MonomPICond}).
\end{proof}

Note that (\ref{eqn:MonMLFSREqn}) is an MLFSR problem 
as in \cite{FengTzeng1989,FengTzeng1991,SchmidtSidorenko2006}. 
The point of Proposition~\ref{proposition:MLFSREqns} 
is the guarantee from the partial-inverse condition.

Instead of using an MLFSR algorithm,
we can solve either of the key equations (\ref{eqn:ReducedKeyEqn}) and (\ref{eqn:MonKeyEqn}) 
by solving SPI problems as shown in Algorithms \ref{alg:SPIErrorLocatingAlgo} 
and~\ref{alg:MonSPIErrorLocatingAlgo}, respectively (shown in framed boxes).
Note that line~\ref{line:ErrLocSPISetTau} initializes $\tau$ to 
$\max_i \{ \deg \tilde m^{(i)}(x) \}$.
Clearly, these algorithms will always terminate, 
and they will return the correct error locator $\Lambda_E(x)$ 
if $\Lambda_E(x)$ satisfies the partial-inverse condition.

\begin{table}[tp]
\framebox[\linewidth]{%
\normalsize%
\begin{minipage}{0.95\linewidth}
\begin{algorithm}[Error Location by SPI Algorithm]\label{alg:SPIErrorLocatingAlgo}\\
{Input:} $S^{(i)}(x)$ for $i=1,\ldots,L$.\\
{Output:} nonzero $\Lambda(x)\in F[x]$, a candidate for\\ 
\phantom{Output:} the error locator $\Lambda_E(x)$ (up to a scale factor).
\begin{pseudocode}
\npcl[line:ErrLocSPISetTau] $\tau := n-k_\text{min}$\\
\npcl \pkw{loop}\\
\npcl[line:callSPI] \> solve the SPI problem%
                \footnote{omitting indices $i$ with $\deg \tilde m^{(i)}(x) < \tau$} 
                with $\tau^{(i)} = \tau$, \\
      \> \> \>  $m^{(i)}(x) = \tilde m^{(i)}(x)$, and $b^{(i)}(x) = S^{(i)}(x)$ \\
\npcl[line:ErrLocSPIBound] \> \pkw{if} $\deg \Lambda(x) > n-k_\text{max}$, \\
      \> \> \> \> \> declare ``decoding failure''\\
\npcl \> \pkw{if} (\ref{eqn:ReducedKeyEqn}) holds, \pkw{return} $\Lambda(x)$\\
\npcl \> $\tau := \tau-1$\\
\npcl \pkw{end}
\end{pseudocode}
\end{algorithm}
\end{minipage}
}
\vspace{\dblfloatsep}

\framebox[\linewidth]{%
\normalsize%
\begin{minipage}{0.95\linewidth}
\begin{algorithm}[Error Location by Monomial-SPI Alg.]\label{alg:MonSPIErrorLocatingAlgo}\\[1ex]
Same as Algorithm~\ref{alg:SPIErrorLocatingAlgo},
but with $\breve S^{(i)}(x)$ instead of $S^{(i)}(x)$,
$x^{n-k^{(i)}}$ instead of $\tilde m^{(i)}(x)$, 
and checking (\ref{eqn:MonKeyEqn}) instead of (\ref{eqn:ReducedKeyEqn}).
\end{algorithm}
\end{minipage}
}
\end{table}

It might seem that Algorithms \ref{alg:SPIErrorLocatingAlgo} 
and~\ref{alg:MonSPIErrorLocatingAlgo} are inefficient. 
However, the SPI algorithms of this paper (cf.\ Section~\ref{sec:SPIAlg}) 
have this property: 
computing the solution for $\tau^{(i)}=\tau=t$ 
is effected by first computing the solution for $\tau^{(i)}=t+1$. 
In other words, 
line~\ref{line:callSPI} of Algorithms \ref{alg:SPIErrorLocatingAlgo} 
and~\ref{alg:MonSPIErrorLocatingAlgo} for $\tau=t$ 
is naturally implemented as continuing the computation for $\tau=t+1$. 
Moreover, the check of (\ref{eqn:ReducedKeyEqn}) 
(or of (\ref{eqn:MonKeyEqn})) can be naturally integrated 
into these algorithms (cf.\ Section~\ref{sec:SPIDecoding}).
When implemented in this way, 
the complexity of Algorithm~\ref{alg:MonSPIErrorLocatingAlgo} is 
$O(L(n-k_\text{min})(n-k_\text{max}))$ additions and/or multiplications in $F$, 
which agrees with the complexity of the MLFSR algorithm of 
\cite{FengTzeng1991,SchmidtSidorenko2006}.

No matter which algorithm is used to produce a candidate 
$\Lambda(x)$ for $\Lambda_E(x)$, 
it may be helpful to check the condition
\begin{equation} \label{eqn:mxMultipleLambdaCheck}
m(x) \bmod \Lambda(x) = 0,
\end{equation}
which guarantees that $\Lambda(x)$ 
is a valid error locator polynomial
(i.e., a product of different factors $(x-\beta_\ell)$, $\ell\in \{ 0,\ldots, n-1\}$, 
up to a scale factor). 
If this condition is not satisfied, then 
decoding failure should be declared.

\section{Successful-Decoding Guarantees and Probabilities}
\label{sec:DecodingFailureAnalysis}

We have seen in Section~\ref{sec:ComputingEL} that decoding 
(by any of the mentioned methods) will certainly succeed if the error pattern $E$
satisfies the partial-inverse condition. 
In this section, we analyze conditions and probabilities for this to happen.
Our main results are Theorems \ref{theorem:IRSSufficientCond} 
and~\ref{theorem:ProboffailuredecodingIRS} below.

\subsection{Roth--Vontobel Bound}
\label{sec:GenRothVontobel}

Theorem~\ref{theorem:IRSSufficientCond} generalizes a result from \cite{RothVontobel2014} 
(paraphrased in (\ref{IRS:guaranteedcorrectingbound}))
from the specific decoding algorithm of \cite{RothVontobel2014}
to any decoding algorithm as in Section~\ref{sec:ComputingEL}.

\begin{theorem}
\label{theorem:IRSSufficientCond}
Let $r_E$ be the rank of the error pattern $E \in F^{L\times n}$ as a matrix over $F$.
If
\begin{equation}\label{eqn:guaranteederrorcottection}
|U_E|\leq \frac{n-k_\text{max}+r_E-1}{2}
\end{equation}
then $E$ satisfies the partial-inverse condition.
\end{theorem}

The partial-inverse condition is thus implied by (\ref{eqn:guaranteederrorcottection}). 
In particular, any nonzero error pattern $E$ with $|U_E| \leq (n-k_\text{max})/2$ satisfies 
the partial-inverse condition.

Note that (\ref{eqn:guaranteederrorcottection}) implies 
$|U_E| < n - k_\text{max}$ (since $r_E \leq |U_E|$).
Theorem~\ref{theorem:IRSSufficientCond} then
follows immediately from the following lemma,
which is inspired by \cite{Metzner,Haslach19992001}
and, especially, \cite{RothVontobel2014}.

\begin{lemma}
Let $E\in F^{L\times n}$ be an error pattern with rank $r_E$ and
\begin{equation} \label{eqn:RankGuaranteeLemmaSupport}
2 |U_E| < n - k_\text{max} + r_E.
\end{equation}
Then any nonzero polynomial $\Lambda(x)$ such that 
\begin{equation} \label{eqn:RankGuaranteeLemmaSPI}
\deg \!\Big( E^{(i)}(x) \Lambda(x) \bmod m(x) \Big) < k^{(i)} + |U_E|
\end{equation} 
for all $i\in \{ 1,\ldots, L \}$ 
is a multiple of $\Lambda_E(x)$.
\end{lemma}

\begin{proof}
The lemma trivially holds for $E=0$. 
We now assume $U_E \neq \emptyset$. 
For any $s\in U_E$, 
let $\tilde U_s$ be a subset of $U_E$ that contains $s$
such that
$|U_E| - |\tilde U_s| = r_E -1$
and, in addition,
the columns $e_\ell$ of $E$ with $\ell \in (U_E \setminus \tilde U_s) \cup \{ s \}$
are linearly independent. (Note that such a set $\tilde U_s$ exists for any $s\in U_E$.)
Let $\tilde e \in F^n$ be a linear combination of rows of $E$ 
such that 
$\tilde e_\ell = 0$ for $\ell \not\in \tilde U_s$
and $\tilde e_s = 1$.
(Note that such $\tilde e$ exists.)
Finally, we define 
$\tilde E_s(x) \eqdef \psi^{-1}(\tilde e)$. 
We thus have $\tilde E_s(\beta_\ell)=0$ for $\ell \not\in\tilde U_s$
and $\tilde E_s(\beta_s)=1$.
Note also that 
$\tilde E_s(x)$ is a linear combination of $E^{(1)}(x),\ldots,E^{(L)}(x)$.

Now assume that some nonzero $\Lambda(x) \in F[x]$ 
satisfies (\ref{eqn:RankGuaranteeLemmaSPI}).
It follows that we also have
\begin{equation} \label{eqn:ProofRankGuaranteeLemmaModDeg}
\deg \!\Big( \tilde E_s(x) \Lambda(x) \bmod m(x) \Big) < k_\text{max} + |U_E|
\end{equation}
for $i=1,\ldots,L$.
We then write 
\begin{equation}
\tilde E_s(x)\Lambda(x)=g(x)m(x)+\tilde E_s(x)\Lambda(x)\bmod m(x)
\end{equation}
according to the division theorem.
But $\tilde E_s(x)$, and thus $\tilde E_s(x)\Lambda(x)$, has at least 
$n - |\tilde U_s|$ 
zeros in the set $\{ \beta_0, \ldots, \beta_{n-1} \}$. 
Using $|\tilde U_s| = |U_E| - r_E + 1$
and (\ref{eqn:RankGuaranteeLemmaSupport}), we obtain
\begin{equation}
n - |\tilde U_s| \geq |U_E| + k_\text{max}.
\end{equation}
Thus $\tilde E_s(x)\Lambda(x)\bmod m(x)$ has at least 
$|U_E| + k$ zeros in the set $\{ \beta_0, \ldots, \beta_{n-1} \}$,
which contradicts (\ref{eqn:ProofRankGuaranteeLemmaModDeg}) unless 
$\tilde E_s(x)\Lambda(x)\bmod m(x) = 0$.
Thus $\tilde E_s(x)\Lambda(x)\bmod m(x) = 0$ 
for all $s\in U_E$. 
It follows that $\Lambda(\beta_\ell)=0$ for all $s\in U_E$,
which implies that $\Lambda(x)$ is a multiple of $\Lambda_E(x)$.
\end{proof}

It is instructive to consider Theorem~\ref{theorem:IRSSufficientCond}
for random errors.
For any fixed $U_E$, 
assume that 
the nonzero columns of $E$ are uniformly 
and independently distributed over all possible nonzero columns. 
In the special case where $|U_E| \leq L$, 
it is then very likely that $E$ has rank $|U_E|$ 
(as quantified by Proposition~\ref{proposition:FullRankProbability} below),
in which case (\ref{eqn:guaranteederrorcottection}) reduces to $|U_E| < n-k_\text{max}$. 
In other words, $E$ is very likely to satisfy the partial-inverse condition
provided that 
\begin{equation} \label{eqn:FullRankLimits}
|U_E| \leq \min \{ L, n-k_\text{max}-1 \}.
\end{equation}
If (\ref{eqn:FullRankLimits}) holds, 
the probability that $E$ does not satisfy the partial-inverse condition
is bounded by
\begin{proposition}[Full-Rank Probability]\label{proposition:FullRankProbability}
Assume $0 < |U_E| \leq L$. If the $|U_E|$ nonzero columns of $E \in F^{L\times n}$ 
are uniformly and independently distributed over $F^L\setminus\{0\}$, then
\begin{equation}
\Pr\!\big(r_E\neq |U_E|\big)<\frac{q^{-L+|U_E|}}{q-1} \label{eqn:FullRankProb}
\end{equation}
with $q=|F|$.
\end{proposition}
This is certainly standard but, 
for the convenience of the reader, 
we give a (short) proof.

\begin{proof}
Assume $0 < |U_E| \leq L$. 
It is easily seen that 
\begin{IEEEeqnarray}{rCl}
\Pr\!\big(r_E=|U_E|\big)
&=&\frac{(q^L-1)(q^L-q)\cdots (q^L-q^{|U_E|-1})}{(q^L-1)^{|U_E|}}\label{prob:frank} \IEEEeqnarraynumspace\\
&=&\frac{(q^L-q)\cdots (q^L-q^{|U_E|-1})}{(q^L-1)^{|U_E|-1}}\label{prob:frank.1}
\end{IEEEeqnarray}
where the numerator of (\ref{prob:frank}) 
is the number of ways of 
picking $|U_E|$ linearly independent column vectors.
The numerator of (\ref{prob:frank.1}) can be written as
\begin{equation}
q^{L(|U_E|-1)}\big(1-q^{-(L-1)}\big)\cdots \big(1-q^{-(L-|U_E|+1)}\big)
\end{equation}
and thus
\begin{IEEEeqnarray}{rCl}
\IEEEeqnarraymulticol{3}{l}{
\Pr\!\big(r_E=|U_E|\big) 
}\nonumber\\\quad
& > & \big(1-q^{-(L-1)}\big)\cdots \big(1-q^{-(L-|U_E|+1)}\big)
       \IEEEeqnarraynumspace\\
& > & 1 - \sum_{i=1}^{|U_E|-1} q^{-(L-i)} \\
& > & 1 - \frac{q^{-(L-|U_E|)}}{q-1}
\end{IEEEeqnarray}
\eproofnegspace
\end{proof}

\subsection{Schmidt--Sidorenko--Bossert Bound}
\label{sec:GenSSB}

We now turn to random errors 
beyond the full-rank case.
Theorem~\ref{theorem:ProboffailuredecodingIRS} generalizes a result 
from \cite{SchSidoBossert2009}
(paraphrased in (\ref{eqn:SchSidoBossert2009bound2})).

\begin{theorem}%
\label{theorem:ProboffailuredecodingIRS}
For $L>1$ and any error support set $U_E$ with $0<|U_E|\leq n-k_\text{max}$,
assume that the $|U_E|$ nonzero columns of $E \in F^{L\times n}$ 
are uniformly and independently distributed over $F^L\setminus\{0\}$.
Then the probability $P_\text{fpi}$ that $E$ does not satisfy the partial-inverse condition 
is bounded by 
\begin{equation} \label{eqn:LambdaeFailureProbBound}
P_\text{fpi}  <  \frac{q^{-L(n-k_\text{avg})+(L+1)|U_E|}}{q-1}  
\end{equation}
\eproofnegspace
\end{theorem}
The proof is given in Section~\ref{secion:proveProboffailuredecodingIRS} below. 
The right side of (\ref{eqn:LambdaeFailureProbBound}) 
agrees with 
the bound (14) of \cite{SchSidoBossert2009}
except for the factor (\ref{eqn:SchSidoBossert2009boundPrefactor}).
(But \cite{SchSidoBossert2009} considers a specific decoding algorithm 
and only cyclic Reed--Solomon codes.)

For $|U_E| \leq L$, both (\ref{eqn:FullRankProb}) and (\ref{eqn:LambdaeFailureProbBound}) apply. 
In general, (\ref{eqn:LambdaeFailureProbBound}) is much stronger than (\ref{eqn:FullRankProb}), 
but the two bounds agree in the special case where $|U_E|=n-k_\text{avg}-1$.

Note that (\ref{eqn:LambdaeFailureProbBound})
implies that $E$ satisfies the partial-inverse condition 
(with high probability, if $q$ is large) 
as long as
\begin{equation}\label{bound:maximumELRadius}
|U_E| \leq
   \min \left\{ 
        n - k_\text{max},\ 
        \frac{L}{L+1}(n-k_\text{avg})
        \right\}.
\end{equation}
This cannot be improved:
\begin{proposition}
(\ref{bound:maximumELRadius}) is a necessary condition 
for $E$ to satisfy the partial-inverse condition.
\end{proposition}
\begin{proof}
Assume that $E$ satisfies the partial-inverse condition,
which means that 
$|U_E|\leq n - k_\text{max}$ and
$\Lambda(x) = \Lambda_E(x)$ is the SPI solution of (\ref{eqn:PICond}).
Applying the degree bound (Proposition~\ref{proposition:SolutionDegree}) 
to (\ref{eqn:PICond}) yields 
\begin{IEEEeqnarray}{rCl}
\deg \Lambda(x) & \leq & \sum_{i=1}^L \left( n - \big( k^{(i)} + |U_E| \big) \right)
                  \IEEEeqnarraynumspace\\
 & = & L \big( n - k_\text{avg} - |U_E| \big)
\end{IEEEeqnarray}
and thus 
$|U_E| (1+L) \leq L (n-k_\text{avg})$.
\end{proof}

For $k^{(1)} = \ldots = k^{(L)} = k_\text{avg}$, 
(\ref{bound:maximumELRadius}) reduces to
\begin{equation}
|U_E| \leq \frac{L}{L+1}(n-k_\text{avg}), 
\end{equation}
which cannot be improved by allowing general $k^{(1)},\ldots,k^{(L)}$.

\subsection{Proof of Theorem~\ref{theorem:ProboffailuredecodingIRS}}
\label{secion:proveProboffailuredecodingIRS}

Let $L>1$ and let $U$ be an arbitary, but fixed, subset of $\{ 0, \ldots, n-1\}$
such that $0<|U|\leq n-k_\text{max}$. 
Assume that the error pattern $E\in F^{L\times n}$ has support set $U_E=U$,
and the nonzero columns of $E$ 
are uniformly and independently distributed over $F^L\setminus\{0\}$.

If $E$ does not satisfy the partial-inverse condition,
then there exists a nonzero polynomial $\Lambda(x)$
with $\deg \Lambda(x) < |U_E|$ such that 
\begin{equation}\label{eqn:IRSerrorevent}
\deg \!\big( E^{(i)}(x) \Lambda(x) \bmod m(x) \big) < k^{(i)} + |U_E|
\end{equation} 
for all $i=1,\ldots,L$.

Let $\calE_U$
be the set of all the possible error patterns $E\in F^{L\times n}$ 
with support set $U_E = U$. 
Let $\calE_\text{fpi} \subset \calE_U$ 
be the set of all $E\in \calE_U$ that admit 
some $\Lambda(x)\in F[x]$ with $0\leq \deg \Lambda(x)<|U|$ 
that satisfies (\ref{eqn:IRSerrorevent}) for all $i\in\{1,\ldots,L\}$.
Then,
\begin{equation} \label{eqn:IRSProofPf1}
P_\text{fpi} \leq \frac{|\calE_\text{fpi}|}{|\calE_U|} = \frac{|\calE_\text{fpi}|}{(q^L-1)^{|U|}}
\end{equation}
It thus remains to bound $|\calE_\text{fpi}|$.

For \mbox{$t=0,\ldots,|U|-1$}, let $\calL_t$ be
the set of monic polynomials $\Lambda(x)\in F[x]$
with $\deg \Lambda(x) < |U|$ and with exactly $t$ zeros 
in the set $\calB_U \eqdef \{ \beta_\ell : \ell\in U \}$.

\begin{lemma} \label{lemma:IRSNumErrorsFixedLambda}
For any fixed $\Lambda(x) \in \calL_t$,
the number of error patterns $E\in \calE_U$ 
that satisfy (\ref{eqn:IRSerrorevent}) is upper bounded by 
$q^{L(2|U|-(n-k_\text{avg})-t)}$.
\end{lemma}
The proof will be given below.
We then have
\begin{equation} \label{eqn:IRSCardSf}
|\calE_\text{fpi}| \leq \sum_{t=0}^{|U|-1} |\calL_t|\, q^{L(2|U|-(n-k_\text{avg})-t)}.
\end{equation}

\begin{lemma} \label{lemma:IRSNumLambdat}
\begin{equation}
|\calL_t| = {|U| \choose t} (q-1)^{|U|-t-1}.
\end{equation}
\eproofnegspace
\end{lemma}
The proof is given below.
Thus (\ref{eqn:IRSCardSf}) becomes
\begin{IEEEeqnarray}{rCl}
|\calE_\text{fpi}|  
 & \leq & \sum_{t=0}^{|U|-1} {|U| \choose t} (q-1)^{|U|-t-1} q^{L(2|U|-(n-k_\text{avg})-t)} 
         \IEEEeqnarraynumspace\\
 & = &w\sum_{t=0}^{|U|-1}  {|U| \choose t}(q-1)^{-t} q^{-Lt}
         \IEEEeqnarraynumspace  \label{eqn:IRSbeginusew}\\
 & < &w \sum_{t=0}^{|U|} {|U| \choose t}\left( (q-1)^{-1} q^{-L} \right)^t
         \IEEEeqnarraynumspace\\
 & = & w\big( 1 + (q-1)^{-1} q^{-L} \big)^{|U|}    \IEEEeqnarraynumspace \label{eqn:IRSendusew}
\\
 & = &  \frac{q^{L(|U|-(n-k_\text{avg}))}}{q-1} \left( (q-1) q^L + 1 \right)^{|U|}
\end{IEEEeqnarray}
with 
\begin{equation}
w \eqdef (q-1)^{|U|-1} q^{L(2|U|-(n-k_\text{avg}))}
\end{equation}
in (\ref{eqn:IRSbeginusew})--(\ref{eqn:IRSendusew}). From (\ref{eqn:IRSProofPf1}), we then have
\begin{IEEEeqnarray}{rCl}
P_\text{fpi}
        & < & \frac{q^{L(|U|-(n-k_\text{avg}))}}{q-1} \left(\frac{q^{L+1}-q^L+1}{q^L-1}\right)^{|U|} \\
        & = &  \frac{q^{-L(n-k_\text{avg}-|U|)+|U|}}{q-1} \left(\frac{q^{L}-(q^{L-1}-q^{-1})}{q^L-1}\right)^{|U|}
       \IEEEeqnarraynumspace
\end{IEEEeqnarray}
and (\ref{eqn:LambdaeFailureProbBound}) follows if $L>1$.

For the proof of Lemma~\ref{lemma:IRSNumErrorsFixedLambda},
we will use the following elementary fact.

\begin{proposition}
\label{proposition:IRSNumPolyWithZeros}
The number of nonzero polynomials over $F$ of degree at most $\nu$
and with $\mu\leq\nu$ prescribed zeros in $F$ (and allowing additional zeros in $F$)
is \mbox{$|F|^{\nu-\mu+1}-1$}.
\end{proposition}

\begin{proofof}{of Lemma~\ref{lemma:IRSNumErrorsFixedLambda}}
Consider the polynomial $E^{(i)}(x) = \psi^{-1}(e^{(i)})$ where $e^{(i)}$ is a row of $E$,
and let $\tilde E^{(i)}(x)\eqdef E^{(i)}(x)\Lambda(x)\bmod m(x)$.
From (\ref{eqn:DefMappingPsi}), we have
\begin{equation} \label{eqn:IRSTildeEIsomorph}
\tilde E^{(i)}(\beta_\ell) = e_{i,\ell} \Lambda(\beta_\ell)
\end{equation}
where $e_{i,\ell}$ denotes the element in row $i$ and column $\ell$ of $E$.
From (\ref{eqn:IRSerrorevent}), we have $\deg \tilde E^{(i)}(x) < k^{(i)}+|U|$.
But (\ref{eqn:IRSTildeEIsomorph})
implies that $\tilde E^{(i)}(x)$ has at least 
$n-|U|+t$ zeros in prescribed positions:
$e_{i,\ell}=0$ for $\ell\not\in U$ and
$\Lambda(x)$ has $t$ zeros in $\calB_U = \{ \beta_\ell: \ell\in U \}$.
By Proposition~\ref{proposition:IRSNumPolyWithZeros},
the number of such polynomials $\tilde E^{(i)}$ is bounded by
$q^{2|U| - (n-k^{(i)}) - t}$,  and 
putting all rows together yields the lemma. 
\end{proofof}

\begin{proofof}{of Lemma~\ref{lemma:IRSNumLambdat}}
Consider nonzero polynomials $\Lambda(x)\in F[x]$
with $\deg \Lambda(x)<|U|$ 
and with $t$ prescribed zeros in $\calB_U$ ($=\{ \beta_\ell: \ell \in U \}$)
and no other zeros in $\calB_U$.
The number of such polynomials $\Lambda(x)$
is $(q-1)^{|U|-t}$,
as is obvious from the ring isomorphism
\begin{equation}
F[x]/m_U(x) \rightarrow F^{|U|} : 
       \Lambda(x) \mapsto \big(\Lambda(\beta_1'),\ldots,\Lambda(\beta'_{|U|})\big)
\end{equation}
with $m_U(x)\eqdef \prod_{\ell\in U}(x-\beta_\ell)$
and $\{ \beta'_1, \ldots, \beta'_{|U|} \} \eqdef \calB_U$.
Lemma~\ref{lemma:IRSNumLambdat} then follows from noting
that it counts only monic polynomials.
\end{proofof}

\section{The Reverse--Berlekamp Massey SPI Algorithm}
\label{sec:SPIAlg}

\begin{table}[tp]
\framebox[\linewidth]{%
\normalsize%
\begin{minipage}{0.95\linewidth}
\begin{algorithm}[Basic SPI Algorithm]\label{alg:BasicSPIAlg}\\
(Reverse Berlekamp--Massey algorithm)\\[0.5ex]
{Input:} $b^{(i)}(x), m^{(i)}(x), \tau^{(i)}$ for $i=1,\ldots,L$.\\
{Output:} $\Lambda(x)$ as in the problem statement.
\begin{pseudocode}
\npcl[line:SPIforinitial] \pkw{for} $i=1,\ldots, L$ \pkw{begin}\\
\npcl \>$\Lambda^{(i)}(x):=0$\\
\npcl[line:SPIinitialdi] \>$d^{(i)}:=\deg m^{(i)}(x)$\\
\npcl \>$\kappa^{(i)}:=\lcf m^{(i)}(x)$\\
\npcl[line:SPIforendinitial] \pkw{end} \\
\npcl[line:SPIinitialLam] $\Lambda(x):=1$\\
\npcl[line:SPIinitialdelta] $\delta:=\max_{i\in\{1,\ldots,L\}}\big(\deg m^{(i)}(x)-\tau^{(i)}\big)$\\
\npcl[line:SPIinitialiddx] $i:=1$\\
\npcl[line:SPIloopbegin]   \pkw{loop begin} \\
\npcl[line:SPIrepeat] \> \pkw{repeat}\\
\npcl[line:SPIupdatei]\> \>  \pkw{if} $i>1$  \pkw{begin} $i:=i-1$ \pkw{end}\\
\npcl[line:SPIupdtatedelta]\> \>  \pkw{else begin} \\
\npcl[line:SPIstop] \>\>\> \pkw{if} $\delta\leq 0$  \pkw{return} $\Lambda(x)$\\
\npcl \>\>\>$i:=L$\\
\npcl \>\>\>$\delta:=\delta-1$\\
\npcl[line:SPIupdateiEnd]\> \> \pkw{end} \\
\npcl[line:SPIupdated]\>\> $d:=\delta+\tau^{(i)}$\\
\npcl[line:SPIkappa] \> \> $\kappa:= \text{coefficient of $x^{d}$ in}$ \\
      \> \> \> \> \> \> $b^{(i)}(x) \Lambda(x) \bmod m^{(i)}(x)$ \\
\npcl[line:SPIuntil] \> \pkw{until} $\kappa \neq 0$\\
\> \> \> {\grey\rule[0.5ex]{50mm}{1pt}}\\
\npcl[line:SPIifswap] \> \pkw{if} $ d<d^{(i)}$ \pkw{begin} \\
\npcl[line:SPIswapbegin] \>\> \pkw{swap} $(\Lambda(x),\Lambda^{(i)}(x))$\\
\npcl[line:SPIswapddi] \>\> \pkw{swap} $(d, d^{(i)})$\\
\npcl[line:SPIswapend] \>\> \pkw{swap} $(\kappa, \kappa^{(i)})$\\
\npcl[line:SPIresetdelta] \>\> $\delta:=d-\tau^{(i)}$ \\
\npcl[line:SPIifswapend] \> \pkw{end} \\
\> \> \> {\grey\rule[0.5ex]{50mm}{1pt}}\\
\npcl[line:SPIupdateLambda] \>
$\Lambda(x):= \kappa^{(i)} \Lambda(x)- \kappa x^{d-d^{(i)}} \Lambda^{(i)}(x)$ \\
\npcl[line:SPIloopend]   \pkw{end}
\end{pseudocode}
\end{algorithm}
See also the refinement in Algorithm~\ref{alg:BasicSPIASpec} below.
\end{minipage}%
}
\vspace{\dblfloatsep}

\framebox[\linewidth]{%
\normalsize%
\begin{minipage}{0.95\linewidth}
\begin{algorithm}[Monomial-SPI Algorithm]\label{alg:BasicSPIASpec}\\[0.5ex]
~In the special case $m^{(i)}(x)=x^{\nu_i}$,\\
line~\ref{line:SPIkappa} of Algorithm~\ref{alg:BasicSPIAlg}
amounts to 
\begin{pseudocode}
%
%
$\kappa :=
\left\{ \begin{array}{ll}
 0, & \text{if $d\geq \nu_i$}\\
 \multicolumn{2}{l}{%
 b_{d}^{(i)} \Lambda_0 + b_{d-1}^{(i)} \Lambda_1 + \ldots + b_{d-s}^{(i)} \Lambda_s,}\\
    & \text{if $d< \nu_i$,}
 \end{array}\right.$
\end{pseudocode}
with $s \eqdef \deg \Lambda(x)$ and $b_\ell^{(i)}  \eqdef 0$ for $\ell<0$. 
\end{algorithm}
\end{minipage}%
}
\end{table}

We now consider algorithms to solve the SPI problem. 
The basic algorithm (the reverse Berlekamp--Massey algorithm)
is stated as Algorithm~\ref{alg:BasicSPIAlg} in the framed box.
The important special case where $m^{(i)}(x) = x^{\nu_i}$ 
is stated as Algorithm~\ref{alg:BasicSPIASpec}. 
(Algorithm~\ref{alg:BasicSPIASpec} is strikingly similar to 
the MLFSR algorithm of \cite{FengTzeng1991} and \cite{SchmidtSidorenko2006}, 
but it is nonetheless a different algorithm.)
Since every SPI problem can be efficiently monomialized, 
Algorithm~\ref{alg:BasicSPIASpec} is the preferred SPI algorithm.

Two variations of Algorithm~\ref{alg:BasicSPIAlg} are given in 
Appendix~\ref{section:QuotientRemainderSavingAlgs}. 
These variations generalize the Quotient Saving Algorithm 
and the Remainder Saving Algorithm (a Euclidean algorithm) 
of \cite{YuLoeliger2016IT}
to $L>1$. However, for $L>1$, these algorithms are less attractive
than Algorithm~\ref{alg:BasicSPIASpec}.

\subsection{Beginning to Explain the Algorithm}
\label{sec:AlgSomeExplanations}

Lines~\ref{line:SPIforinitial}--\ref{line:SPIinitialiddx} 
of Algorithm~\ref{alg:BasicSPIAlg}
are for
initialization; 
the nontrivial part begins with line~\ref{line:SPIloopbegin}.
Note that lines
\ref{line:SPIswapbegin}--\ref{line:SPIswapend} simply swap
$\Lambda(x)$ with $\Lambda^{(i)}(x)$, $d$ with $d^{(i)}$, and
$\kappa$ with $\kappa^{(i)}$. The only actual computations are in
lines~\ref{line:SPIkappa} and \ref{line:SPIupdateLambda}. 

We now begin to explain the algorithm 
(but the actual proof of correctness will be deferred to 
Appendix~\ref{section:proofAlgo}).
To this end, we define the following quantities.
For any nonzero
$\Lambda(x)$ and any $i\in \{1,2,\ldots,L\}$, let
\begin{equation} \label{eqn:rdi}
\rd^{(i)}(\Lambda)\eqdef \deg\!\big(b^{(i)}(x)\Lambda(x)\bmod m^{(i)}(x)\big),
\end{equation}
\begin{equation} \label{eqn:deltamax}
\delta_{\text{max}}(\Lambda)\eqdef\max_{i\in \{1,\ldots,L\}}
\Big(\rd^{(i)}(\Lambda)-\tau^{(i)}\Big),
\end{equation}
and
\begin{equation} \label{eqn:imax}
i_\text{max}(\Lambda) \eqdef \max \argmax_{i\in
\{1,\ldots,L\}}\Big(\rd^{(i)}(\Lambda)-\tau^{(i)}\Big),
\end{equation}
the largest among the indices $i$ that maximize
$\rd^{(i)}(\Lambda)-\tau^{(i)}$, cf.\ Figure~1.

At any given time, the algorithm works on the polynomial
$\Lambda(x)$. %
The inner \pkw{repeat} loop
(lines~\ref{line:SPIrepeat}--\ref{line:SPIuntil})
computes the quantities defined in (\ref{eqn:rdi})--(\ref{eqn:imax}):
between lines~\ref{line:SPIuntil} and \ref{line:SPIifswap}, we have
\begin{equation} \label{eqn:repeatloopaimstodo}
i=i_{\text{max}}(\Lambda),~\delta=\delta_{\text{max}}(\Lambda),~
d=\rd^{(i)}(\Lambda),
\end{equation}
and also
\begin{equation} \label{eqn:repeatloopaimstodo.1}
\kappa=\lcf\!\big(b^{(i)}(x)\Lambda(x)\bmod m^{(i)}(x)\big).
\end{equation}
In particular, the very first execution of the \pkw{repeat} loop
(with $\Lambda(x)=1$)
yields
\begin{equation}
i=\max \argmax_{i\in \{1,\ldots,L\}}\Big(\deg
b^{(i)}-\tau^{(i)}\Big),
\end{equation}
$d=\deg b^{(i)}(x)$, and $\kappa=\lcf b^{(i)}(x)$  between
lines~\ref{line:SPIuntil} and \ref{line:SPIifswap}. 

In the special case $L=1$,
lines~\ref{line:SPIupdatei}--\ref{line:SPIupdated} (excluding
line~\ref{line:SPIstop}) amount to $d:=d-1$;
in this case, the algorithm reduces
to the partial-inverse algorithm of 
\cite{YuLoeliger2016IT}.

The only exit from the algorithm is line~\ref{line:SPIstop}. Since
$\delta \geq \delta_\text{max}(\Lambda)$, the condition $\delta
\leq 0$ guarantees that $\Lambda(x)$
satisfies~(\ref{eqn:SPI}).

\begin{figure}[t!]
\begin{center}
\setlength{\unitlength}{0.8mm}
\begin{picture}(70,50)(5,5)
\put(18,10){\vector(0,1){38}}
\put(18,50){\pos{cb}{$\rd^{(i)}(\Lambda) - \tau^{(i)}$}}
\put(18,40){\line(-1,0){1}}  \put(16,40){\pos{r}{$\delta_\text{max}$}}
\put(18,32){\line(-1,0){1}}  \put(16,32){\pos{r}{$\delta_\text{max}-1$}}
\put(18,10){\vector(1,0){55}}
\put(18,40){\line(1,0){20}}
\put(28,10){\line(0,1){30}}
\put(38,10){\line(0,1){30}}
\put(38,32){\line(1,0){30}} 
\put(48,10){\line(0,1){22}}
\put(58,10){\line(0,1){22}}
\put(68,10){\line(0,1){22}}
\put(23,8){\pos{t}{$1$}}
\put(33,8){\pos{t}{$i_\text{max}$}}
\put(43,8){\pos{t}{$3$}}
\put(53,8){\pos{t}{$4$}}
\put(63,8){\pos{t}{$5$}}
\put(75,8){\pos{t}{$i$}}
\end{picture}
\caption{\label{Fig:1}%
Illustration of (\ref{eqn:deltamax}) and
(\ref{eqn:imax}) for $i_\text{max}=2$.
}
\end{center}
\vspace{-0.2cm}
\end{figure}

The algorithm maintains the auxiliary polynomials
$\Lambda^{(i)}(x)$, $i=1,\ldots, L$, which are all initialized to
$\Lambda^{(i)}(x)=0$. Thereafter, however, $\Lambda^{(i)}(x)$
become nonzero (after their first respective execution of lines
\ref{line:SPIswapbegin}--\ref{line:SPIswapend}) and satisfy
\begin{equation}
i_{\text{max}}(\Lambda^{(i)})=i.
\end{equation}

The heart of the algorithm is line~\ref{line:SPIupdateLambda},
which cancels the leading term in
\begin{equation} \label{eqn:leadingtermcalcelling}
b^{(i)}(x)\Lambda(x)\bmod m^{(i)}(x)
\end{equation}
(except for the first execution for each index $i$,
see below).
Line~\ref{line:SPIupdateLambda} is explained by the
following lemma.
\begin{lemma}[Remainder Decreasing Lemma]\label{lemma:SPIAlgcore}
Let $\Lambda'(x)$ and $\Lambda''(x)$ be nonzero polynomials such
that $i\eqdef i_\text{max}(\Lambda')=i_\text{max}(\Lambda'')$ and
$\rd^{(i)}(\Lambda')\geq \rd^{(i)}(\Lambda'')$. Then
$\delta_\text{max}(\Lambda')\geq \delta_\text{max}(\Lambda'')$ and
the polynomial
\begin{equation} \label{eqn:2WrongsMakeRight}
\Lambda(x) \eqdef \kappa'' \Lambda'(x)- \kappa' x^{d'-d''}
\Lambda''(x)
\end{equation}
with $d'\eqdef \rd^{(i)}(\Lambda')$, $\kappa'\eqdef
\lcf(b^{(i)}(x)\Lambda'(x)\bmod m^{(i)}(x))$, $d''\eqdef
\rd^{(i)}(\Lambda'')$, and
$\kappa''\eqdef\lcf(b^{(i)}(x)\Lambda''(x)\bmod m^{(i)}(x))$
satisfies both
\begin{equation} \label{eqn:rdsmaller}
\rd^{(i)}(\Lambda)<\rd^{(i)}(\Lambda')
\end{equation}
and
\begin{equation}
\delta_\text{max}(\Lambda) \leq  \delta_\text{max}(\Lambda')
\label{eqn:dmaxsmallerorequal}
\end{equation}
and either
\begin{equation}
i_\text{max}(\Lambda)<  i_\text{max}(\Lambda'),
\label{eqn:imaxsmaller}
\end{equation}
or
\begin{equation}
\delta_\text{max}(\Lambda) <  \delta_\text{max}(\Lambda').
\label{eqn:dmaxsmaller}
\end{equation}
\eproofnegspace
\end{lemma}
\begin{proof}
First,
$\delta_\text{max}(\Lambda')\geq \delta_\text{max}(\Lambda'')$ is
obvious from the assumptions. 
For the rest of proof, we define
for every $\ell\in \{1,\ldots,L\}$
\begin{IEEEeqnarray}{rCl}
r'^{(\ell)}(x) & \eqdef & b^{(\ell)}(x) \Lambda'(x) \bmod m^{(\ell)}(x) \label{eqn:CoreLemmaR1} \IEEEeqnarraynumspace\\
r''^{(\ell)}(x) & \eqdef & b^{(\ell)}(x) \Lambda''(x) \bmod m^{(\ell)}(x)
\label{eqn:CoreLemmaR2}
\end{IEEEeqnarray}
and we obtain from (\ref{eqn:2WrongsMakeRight})
\begin{IEEEeqnarray}{rCl}
r^{(\ell)}(x)  & \eqdef &  b^{(\ell)}(x) \Lambda(x) \bmod m^{(\ell)}(x) \IEEEeqnarraynumspace\\
  & = & \kappa'' r'^{(\ell)}(x) - \kappa' x^{d'-d''} r''^{(\ell)}(x) \label{eqn:2WrongsMakeRightRemainder}
\end{IEEEeqnarray}
by the natural ring homomorphism $F[x] \rightarrow
F[x]/m^{(\ell)}(x)$. 
Moreover, we define 
\begin{IEEEeqnarray}{rCl}
\delta^{(\ell)}(\Lambda)&\eqdef& \rd^{(\ell)}(\Lambda)-\tau^{(\ell)} \IEEEeqnarraynumspace\\
\delta^{(\ell)}(\Lambda')&\eqdef& \rd^{(\ell)}(\Lambda')-\tau^{(\ell)} \\
\delta^{(\ell)}(\Lambda'')&\eqdef& \rd^{(\ell)}(\Lambda'')-\tau^{(\ell)}
\end{IEEEeqnarray}
(cf.\ Figure~1 for $\Lambda$, $\Lambda'$, and $\Lambda''$, respectively).
By the stated assumptions, we have for $\ell=i$
\begin{equation} \label{proof:degsum}
\delta^{(i)}(\Lambda')=d'-d''+\delta^{(i)}(\Lambda''),
\end{equation}
and we obtain from (\ref{eqn:2WrongsMakeRightRemainder})
\begin{equation}\label{proof:1fordmaximax}
\delta^{(i)}(\Lambda)<\delta^{(i)}(\Lambda'),
\end{equation}
\begin{equation} \label{proof:2fordmaximax}
\delta^{(\ell)}(\Lambda) \leq  \delta^{(i)}(\Lambda') \text{~for~}
\ell<i,
\end{equation}
and
\begin{equation} \label{proof:3fordmaximax}
\delta^{(\ell)}(\Lambda) <\delta^{(i)}(\Lambda') \text{~for~} \ell>i.
\end{equation}
Clearly, (\ref{proof:1fordmaximax}) implies (\ref{eqn:rdsmaller}); 
(\ref{proof:1fordmaximax})--(\ref{proof:3fordmaximax})
together imply both (\ref{eqn:dmaxsmallerorequal}) and either
(\ref{eqn:imaxsmaller}) or (\ref{eqn:dmaxsmaller}) (or both).
\end{proof}

It follows from (\ref{eqn:rdsmaller})--(\ref{eqn:dmaxsmaller})
that the algorithm makes progress and eventually terminates.

For each index~$i\in \{1,\ldots,L\}$,
when line~\ref{line:SPIupdateLambda}
is executed for the very first time,
it is
necessarily preceded by the swap in
lines~\ref{line:SPIswapbegin}--\ref{line:SPIswapend}. In this
case, line~\ref{line:SPIupdateLambda} reduces to
\begin{equation} \label{eqn:specialcaselambda}
\Lambda(x):=-\Big(\lcf m^{(i)}(x)\Big)x^{\deg
m^{(i)}(x)-\rd^{(i)}(\Lambda')}\Lambda'(x)
\end{equation}
where $\Lambda'(x)$ is the value of $\Lambda(x)$ before the swap.
It follows, in particular, that $\deg \Lambda(x)>\deg
\Lambda'(x)$.

In any case, we always have
\begin{equation}\label{algorithm:rdiLambda}
\deg\!\big(b^{(i)}(x)\Lambda(x)\bmod m^{(i)}(x)\big)<d
\end{equation}
after executing line~\ref{line:SPIupdateLambda}.

Finally, we note that every execution of the swap in
lines~\ref{line:SPIswapbegin}--\ref{line:SPIswapend} strictly
reduces $d^{(i)}$.
We also note that the execution of line~\ref{line:SPIresetdelta}
results in
\begin{equation} \label{eqn:deltavalue}
\delta= \left\{ \begin{array}{ll}
         \delta_\text{max}(\Lambda), & \text{if $\Lambda(x)\neq 0$} \\
         \deg m^{(i)}-\tau^{(i)}, & \text{if $\Lambda(x)= 0$,}
        \end{array}\right.
\end{equation}
where the second case happens only once---the very first time---%
for each index $i\in \{ 1,\ldots, L \}$.

\begin{theorem}
\label{theorem:CorrectnessOfBasicAlg}
Algorithm~\ref{alg:BasicSPIAlg} returns the solution of the SPI problem.
\end{theorem}
The proof will be given in 
Appendix~\ref{section:proofAlgo}.

\subsection{Complexity of Algorithm~\ref{alg:BasicSPIAlg}}

Let 
\begin{equation}\label{def:Bbound}
\maxD \eqdef L \max_{i\in\{1,\ldots,L\}}\!\big(\deg m^{(i)}(x)-\tau^{(i)}\big).
\end{equation}
and note that $D$ as in (\ref{eqn:DegreeBoundsymbol}) satisfies 
$D\leq \maxD$.
\begin{theorem}[Number of Iterations]\label{theorem:NumIterations}
The number $N_\text{it}$ of executions of line~\ref{line:SPIkappa} 
of Algorithm~\ref{alg:BasicSPIAlg} is
\begin{IEEEeqnarray}{rCl}
N_\text{it} &=& \maxD + L\deg\Lambda(x)\label{Nit:ExactBound}\\
&\leq & \maxD + LD \label{Nit:UpperBound}
\end{IEEEeqnarray}
where $\Lambda(x)$ is the solution of the SPI problem.
\end{theorem}

The step from (\ref{Nit:ExactBound}) to (\ref{Nit:UpperBound}) is obvious from 
(\ref{eqn:MaxDegree}). The proof of
(\ref{Nit:ExactBound}) will be given in 
Appendix~\ref{sec:ProofOftheorem:NumIterations}.

In the special case addressed by Algorithm~\ref{alg:BasicSPIASpec},
with $m^{(i)}(x) = x^{\nu_i}$ for $i=1,\ldots,L$ and 
$\nu_\text{max} \eqdef \max_{i \in \{1,\ldots,L\}} \nu_i$, we have
\begin{theorem}
The complexity of Algorithm~\ref{alg:BasicSPIASpec} is bounded by
\begin{equation} \label{eqn:MonomialRevBMAComplexity}
O(N_\text{it} \deg \Lambda(x)) \leq
  O \min\!\left( 
  \maxD D + L D^2,\mbox{~}
  L \nu_\text{max}^2
\right)
\end{equation}
additions and multiplications in $F$.
\end{theorem}

\begin{proof}
The left side of (\ref{eqn:MonomialRevBMAComplexity})
is obvious from Algorithm~\ref{alg:BasicSPIASpec}. 
The first term on the right side 
follows from (\ref{Nit:UpperBound}) and Proposition~\ref{proposition:SolutionDegree}.
The second term on the right side 
follows from (\ref{Nit:ExactBound}), $\maxD \leq L\nu_\text{max}$, and
Proposition~\ref{proposition:MonomialDegreeBound}. 
\end{proof}

The complexity for decoding (as in Algorithm~\ref{alg:SPIErrorLocatingAlgo})
will be addressed in Section~\ref{sec:SPIDecoding}.

\section{Using the SPI Algorithm for Decoding}
\label{sec:SPIDecoding}

\begin{table}[tp]
\framebox[\linewidth]{%
\normalsize%
\begin{minipage}{0.95\linewidth}
\begin{algorithm}[Basic SPI Error-Locating Algorithm]\label{alg:SPIErrorLocatingAlg}\\
(an implementation of Algorithm~\ref{alg:SPIErrorLocatingAlgo})\\[1ex]
{Input:} $S^{(i)}(x)$, $\tilde m^{(i)}(x)$, $n$, and $k^{(i)}$\\
{Output:} nonzero \mbox{$\Lambda(x)\in F[x]$}, same as Algorithm~\ref{alg:SPIErrorLocatingAlgo}.%
\vspace{1ex}

Use Algorithm~\ref{alg:BasicSPIAlg}
with $b^{(i)}(x)=S^{(i)}(x)$, $m^{(i)}(x)=\tilde m^{(i)}(x)$, $\tau^{(i)} = n-k_\text{min}$, 
and with the following 
three modifications: 
first, \mbox{$\tau^{(i)} = \tau$} does not depend on $i$ 
(but it is decreased during the algorithm, see below);
second, initialize also \mbox{$d:= n-k_\text{min}$};
third, replace line~\ref{line:SPIstop} of Algorithm~\ref{alg:BasicSPIAlg}
with the following lines:
\begin{pseudocode}
\setcounter{proglinecounter}{70}
\npcl \pkw{if} $\delta\leq 0$ \pkw{begin} \\
\npcl[line:SPIErrLocFail] \> \pkw{if} $\deg \Lambda(x) > n-k_\text{max}$ \pkw{return} ``decod.\ failure''\\
\npcl[GEL:stopcond]\> \pkw{if} $d\leq\deg \Lambda(x)$  \pkw{return} $\Lambda(x)$\\
\npcl\> \pkw{else begin}\\
\npcl\>\> $\tau:=\tau-1$\\
\npcl\>\> $\delta:=\delta+1$\\
\npcl\> \pkw{end}\\
\npcl\pkw{end}
\end{pseudocode}
\end{algorithm}
\end{minipage}
}
\vspace{\dblfloatsep}

\framebox[\linewidth]{%
\normalsize%
\begin{minipage}{0.95\linewidth}
\begin{algorithm}[Monomial-SPI Error-Locating Algorithm]\label{alg:MonSPIErrorLocatingAlg}\\
(an implementation of Algorithm~\ref{alg:MonSPIErrorLocatingAlgo})\\[1ex]
{Input:} $\breve S^{(i)}(x)$, $n$, and $k^{(i)}$\\
{Output:} nonzero \mbox{$\Lambda(x)\in F[x]$}, same as Algorithm~\ref{alg:MonSPIErrorLocatingAlgo}.%
\vspace{1ex}

The algorithm is Algorithm~\ref{alg:BasicSPIASpec} 
with $b^{(i)}(x) = \breve S^{(i)}(x)$, $m^{(i)}(x)=x^{n-k^{(i)}}$, $\tau^{(i)} = n-k_\text{min}$,
and with modifications as in Algorithm~\ref{alg:SPIErrorLocatingAlg}.
\end{algorithm}
\end{minipage}
}
\end{table}

\begin{table}[tp]
\framebox[\linewidth]{%
\normalsize%
\begin{minipage}{0.95\linewidth}
\begin{algorithm}[Fixed-Iterations Algorithm]\label{Altalg:SPIErrorLocatingAlg}\\
{Input:} $S^{(i)}(x)$, $\tilde m^{(i)}(x)$, $n$, and $k^{(i)}$\\
{Output:} nonzero \mbox{$\Lambda(x)\in F[x]$}, same as%
\footnote{Except that the condition $\deg \Lambda(x) \leq n - k_\text{max}$
is not checked inside the algorithm, but should be added as an external check.}
Algorithm~\ref{alg:SPIErrorLocatingAlgo}.%
\vspace{1ex}

The algorithm is Algorithm~\ref{alg:BasicSPIAlg} 
with \mbox{$b^{(i)}(x)=S^{(i)}(x)$}, \mbox{$m^{(i)}(x)=\tilde m^{(i)}(x)$}, 
\mbox{$\tau^{(i)} = 0$},
and with the following modifications:
first, there is an extra integer variable $N_\text{it}$ 
that is initialized to zero;
second, line~\ref{line:SPIstop} is replaced with
\begin{pseudocode}
\setcounter{proglinecounter}{80}
\npcl[nitnk] \pkw{if} $N_\text{it} = L(n-k_\text{min})$ \pkw{return} $\Lambda(x)$
\end{pseudocode}
and third, the extra line
\begin{pseudocode}
\setcounter{proglinecounter}{90}
\> \>  $N_\text{it}:=N_\text{it}+1$
\end{pseudocode}
is inserted between lines~\ref{line:SPIkappa} and \ref{line:SPIuntil}.

\end{algorithm}
\end{minipage}
}
\vspace{\dblfloatsep}

\framebox[\linewidth]{%
\normalsize%
\begin{minipage}{0.95\linewidth}
\begin{algorithm}[Monomial-SPI Fixed-Iterations Algorithm]\label{Altalg:MonSPIErrorLocatingAlg}\\
{Input:} $\breve S^{(i)}(x)$, $n$, and $k^{(i)}$\\
{Output:} nonzero \mbox{$\Lambda(x)\in F[x]$}, same as%
\footnote{See the footnote in Algorithm~\ref{Altalg:SPIErrorLocatingAlg}.}
Algorithm~\ref{alg:MonSPIErrorLocatingAlgo}.%
\vspace{1ex}

The algorithm is Algorithm~\ref{alg:BasicSPIASpec} with 
$b^{(i)}(x) = \breve S^{(i)}(x)$, $m^{(i)}(x)=x^{n-k^{(i)}}$, $\tau^{(i)} = 0$,
and with modifications as in Algorithm~\ref{Altalg:SPIErrorLocatingAlg}.
\end{algorithm}
\end{minipage}
}
\end{table}

As described in Section~\ref{sec:ComputingEL},
the SPI algorithm can be used to compute (an estimate of) 
the error locator polynomial of an interleaved 
Reed--Solomon code as in Section~\ref{sec:SPIProblem4Decoding}.
The preferred version of such a decoding algorithm is Algorithm~\ref{Altalg:MonSPIErrorLocatingAlg},
which is the final result of this section.
We get there step by step, beginning with Algorithm~\ref{alg:SPIErrorLocatingAlg}
(see the framed boxes).

Algorithm~\ref{alg:SPIErrorLocatingAlg} implements 
Algorithm~\ref{alg:SPIErrorLocatingAlgo} of Section~\ref{sec:ComputingEL}
using Algorithm~\ref{alg:BasicSPIAlg}, 
and 
Algorithm~\ref{alg:MonSPIErrorLocatingAlg} 
implements Algorithm~\ref{alg:MonSPIErrorLocatingAlgo}
using Algorithm~\ref{alg:BasicSPIASpec}.

Line~\ref{GEL:stopcond} makes sure that Algorithm~\ref{alg:SPIErrorLocatingAlg}
stops only when (\ref{eqn:ReducedKeyEqn}) is satisfied for all $i\in \{1,\ldots,L\}$: 
when this line is executed, we always have $d = \tau$ and
\begin{equation}
\deg\big( b^{(i)}(x) \Lambda(x) \bmod \tilde m^{(i)}(x) \big) < d.
\end{equation}
Because line~\ref{GEL:stopcond} checks the condition $d\leq \deg\Lambda(x)$ 
rather than (\ref{eqn:ReducedKeyEqn}), 
Algorithm~\ref{alg:SPIErrorLocatingAlg} may terminate later 
(with a smaller value of $\tau$) 
than Algorithm~\ref{alg:SPIErrorLocatingAlgo}.
In fact, from the moment where (\ref{eqn:ReducedKeyEqn}) holds 
(and assuming that the partial-inverse condition is satisfied),
Algorithm~\ref{alg:SPIErrorLocatingAlg} continues to decrease both $d$ and $\tau$
(without changing $\Lambda(x)=\Lambda_E(x)$)
until $d=|U_E|$.
\begin{lemma} \label{lemma:SPIErrorLocNit}
Assume that $E$ satisfies the partial-inverse condition.
Then Algorithm~\ref{alg:SPIErrorLocatingAlg} stops with $\tau = |U_E|$.
Moreover, the number $N_\text{it}$ of executions of line~\ref{line:SPIkappa}
of Algorithm~\ref{alg:BasicSPIAlg} is $L(n-k_\text{min})$.
\end{lemma}
\begin{proof}
The first claim follows from the discussion above.
From Theorem~\ref{theorem:NumIterations},
we have 
\begin{equation}
N_\text{it} = \hat D + L \deg \Lambda_E(x).
\end{equation}
But $\hat D=L(n-k_\text{min}-\tau)$ with $\tau=|U_E|$ when the algorithm stops.
Thus $N_\text{it} = L(n-k_\text{min})$.
\end{proof}
An immediate consequence is
\begin{proposition}[Monom.-SPI Error Locating Complexity] \label{proposition:SPIErrorLocComplexity}
Assume that $E$ satisfies the partial-inverse condition.
Then the complexity of Algorithm~\ref{alg:MonSPIErrorLocatingAlg} 
is $O\big( L(n-k_\text{min})(n-k_\text{max}) \big)$ additions and multiplications in $F$.
\end{proposition}

In Lemma~\ref{lemma:SPIErrorLocNit} and Proposition~\ref{proposition:SPIErrorLocComplexity},
the conditioning on the partial-inverse condition is unsatisfactory.
But Lemma~\ref{lemma:SPIErrorLocNit} suggests 
a solution to this problem: if $N_\text{it}$ exceeds $L(n-k_\text{min})$,
the partial-inverse condition is not satisfied and it is pointless to continue.
In other words, we can use the condition $N_\text{it} = L(n-k_\text{min})$,
rather than line~\ref{GEL:stopcond}, to stop the algorithm. 
The resulting error-locating algorithm
(as a modification of Algorithm~\ref{alg:BasicSPIAlg})
is given as Algorithm~\ref{Altalg:SPIErrorLocatingAlg} (see the box).
The same modification can also be applied to 
Algorithm~\ref{alg:MonSPIErrorLocatingAlg},
resulting in Algorithm~\ref{Altalg:MonSPIErrorLocatingAlg}.
We then have
\begin{proposition}
The complexity of Algorithm~\ref{Altalg:MonSPIErrorLocatingAlg}
(for arbitrary input) 
is $O\big( L(n-k_\text{min})(n-k_\text{max}) \big)$ additions and multiplications in $F$.
\end{proposition}

\section{Conclusion}
\label{section:conclusion} 

We have introduced the SPI problem for polynomials 
and used it to generalize and to harmonize a number of ideas 
from the literature on decoding interleaved Reed--Solomon codes 
beyond half the minimum distance.
The SPI problem has a unique solution 
(up to a scale factor), which can be computed by 
a (new) multi-sequence reverse Berlekamp--Massey algorithm. 

The SPI problem with general moduli can always (and efficiently) 
be reduced to an SPI problem with monomial moduli. 
For monomial moduli, the reverse Berlekamp--Massey algorithm
looks very much like (and has the same complexity as)
the multi-sequence Berlekamp--Massey algorithm of 
\cite{FengTzeng1991,SchmidtSidorenko2006}.

The SPI problem can be used to analyze 
syndrome-based decoding of interleaved Reed--Solomon codes.
Specifically, we pointed out a natural partial-inverse condition 
for the error pattern, 
which is always satisfied up to half the minimum distance 
and very likely to be satisfied almost up to the full minimum distance. 
If that condition is satisfied, the (true) error locator polynomial 
is the unique solution of a standard key equation 
and can be computed in many different ways,
including 
the algorithm of \cite{FengTzeng1991,SchmidtSidorenko2006} 
and the reverse Berlekamp--Massey algorithm of this paper.
Two of the best performance bounds (for two different decoding algorithms) 
from the literature
were rederived and generalized so that they apply 
to the partial-inverse condition,
and thus simultaneously to many different decoding algorithms.

In Appendix~\ref{section:QuotientRemainderSavingAlgs},
we also give two easy variations of the reverse Berlekamp--Massey algorithm, 
one of which is a Euclidean algorithm. 
However, for $L>1$, these variations have higher complexity
than the reverse Berlekamp--Massey algorithm with monomial moduli.

\appendices

\section{Proof of the SPI Algorithm}
\label{section:proofAlgo}

In this appendix, we prove Theorems \ref{theorem:CorrectnessOfBasicAlg} 
and~\ref{theorem:NumIterations}.

\subsection{Assertions (Properties of the Algorithm)}

To prove the correctness of Algorithm~\ref{alg:BasicSPIAlg}, we augment it
with some extra variables and some assertions 
as shown in Algorithm~\ref{alg:AnnotatedSPIAlg}.
We will prove these assertions one by one, 
except that the proof of 
Assertion~(\assertref{SPIassert:Lamdeginvariant}) 
is deferred to the end of this section.

\begin{table}[tp]
\framebox[\linewidth]{%
\normalsize%
\begin{minipage}{0.95\linewidth}
\begin{algorithm}[Annotated SPI Algorithm]\label{alg:AnnotatedSPIAlg}
\begin{pseudocode}
\npcl \pkw{for} $i=1,\ldots, L$ \pkw{begin}\\
\npcl \>$\Lambda^{(i)}(x):=0$\\
\npcl \>$d^{(i)}:=\deg m^{(i)}(x)$\\
\npcl \>$\kappa^{(i)}:=\lcf m^{(i)}(x)$\\
\npcl \pkw{end} \\
\npcl $\Lambda(x):=1$\\
\npcl $\delta:=\max_{i\in\{1,\ldots,L\}}\big(\deg m^{(i)}(x)-\tau^{(i)}\big)$\\
\npcl $i:=1$\\
\> \> \framebox[0.89\linewidth]{\begin{minipage}{0.86\linewidth}%
           \pkw{Extra:}\\
            $k:=0$   \extralabel{SPIextra:k}
          \end{minipage}}\\[0.5ex]
\npcl   \pkw{loop begin} \\

\> \> \> \framebox[0.83\linewidth]{\begin{minipage}{0.80\linewidth}%
           \pkw{Assertions:}\\
            $\deg \Lambda(x) =\sum_{i=1}^L \big(\deg m^{(i)}(x)-d^{(i)}\big)$
            \assertlabel{SPIassert:Lamdeginvariant}\\
            $\deg \Lambda(x) > \deg \Lambda^{(i)}(x),~~i=1,\ldots,L$ \assertlabel{SPIassert:degLamLLami}
          \end{minipage}}\\[0.5ex]

\npcl \> \pkw{repeat}\\
\npcl \> \>  \pkw{if} $i>1$  \pkw{begin} $i:=i-1$ \pkw{end}\\
\npcl \> \>  \pkw{else begin} \\
\npcl \>\>\> \pkw{if} $\delta\leq 0$  \pkw{return} $\Lambda(x)$\\
\npcl \>\>\>$i:=L$\\
\npcl \>\>\>$\delta:=\delta-1$\\
\npcl \> \> \pkw{end} \\
\npcl \>\> $d:=\delta+\tau^{(i)}$\\
\npcl \> \> $\kappa:= \text{coefficient of $x^{d}$ in}$ \\
      \> \> \> \> \> \> $b^{(i)}(x) \Lambda(x) \bmod m^{(i)}(x)$ \\
\npcl \> \pkw{until} $\kappa \neq 0$\\
\> \> \> \framebox[0.83\linewidth]{\begin{minipage}{0.80\linewidth}%
           \pkw{Assertion:}\\
            $i=i_\text{max}(\Lambda)$,$~\delta=\delta_\text{max}(\Lambda)\geq 0$
            \assertlabel{SPIassert:imaxLambda}
          \end{minipage}}\\[0.5ex]
\npcl \> \pkw{if} $ d<d^{(i)}$ \pkw{begin} \\
\> \> \> \> \framebox[0.765\linewidth]{\begin{minipage}{0.73\linewidth}%
            \pkw{Assertion:}\\
            $d^{(i)}> d = \delta + \tau^{(i)} \geq \tau^{(i)}$
            \assertlabel{SPIassert:diLdLtaui}\\
           \pkw{Extras:}\\
            $k:=k+1$,$~i_k\eqdef i$,$~\Lambda_k(x)\eqdef \Lambda(x)$, \\
            $\Delta_k\eqdef d^{(i)}-d$,$~d_k\eqdef d^{(i)}$  \extralabel{SPIextra:dk}
          \end{minipage}}\\[0.5ex]
\npcl \>\> \pkw{swap} $(\Lambda(x),\Lambda^{(i)}(x))$\\
\npcl \>\> \pkw{swap} $(d, d^{(i)})$\\
\npcl \>\> \pkw{swap} $(\kappa, \kappa^{(i)})$\\
\npcl \>\> $\delta:=d-\tau^{(i)}$ \\
\> \> \> \> \framebox[0.765\linewidth]{\begin{minipage}{0.73\linewidth}%
           \pkw{Assertions:}\\
           $d > d^{(i)} \geq \tau^{(i)}$      \assertlabel{SPIassert:d1Ld2} \\
           $\deg \Lambda^{(i)}(x)>\deg \Lambda(x)$
           \assertlabel{SPIassert:degLamiLLam}\\
           $\deg \Lambda^{(i)}(x) >\deg \Lambda^{(j)}(x) \text{~for~} j\neq i$ \assertlabel{SPIassert:degLambdai}\\
           $i_\text{max}(\Lambda^{(i)})=i$,$~\delta_\text{max}(\Lambda^{(i)})\geq 0$    \assertlabel{SPIassert:imaxLambdai}
          \end{minipage}}\\[0.5ex]
\npcl \> \pkw{end} \\
\npcl  \>
$\Lambda(x):= \kappa^{(i)} \Lambda(x)- \kappa x^{d-d^{(i)}} \Lambda^{(i)}(x)$ \\
\> \> \> \framebox[0.83\linewidth]{\begin{minipage}{0.80\linewidth}%
           \pkw{Assertions:}\\
            $\rd^{(i)}(\Lambda) < d =  \delta + \tau^{(i)}$ \assertlabel{SPIassert:rdilessd}\\
            $\deg \Lambda(x)= \Delta_k+\deg \Lambda_k(x)$
            \assertlabel{SPIassert:altLamdeginvariant}\\
            \phantom{$\deg \Lambda(x)$} $>\deg \Lambda^{(i)}(x),~~i=1,\dots,L$ \assertlabel{SPIassert:degLamLLami.1}
          \end{minipage}}\\[0.5ex]
\npcl   \pkw{end}
\end{pseudocode}
\end{algorithm}
\end{minipage}
}
\end{table}

Assertion (\assertref{SPIassert:degLamLLami}) is obvious both from the
initialization and from (\assertref{SPIassert:degLamLLami.1}).
Assertion (\assertref{SPIassert:imaxLambda}) 
is the result of the \pkw{repeat} loop,
as discussed at the beginning of Section~\ref{sec:AlgSomeExplanations}.

Assertion (\assertref{SPIassert:diLdLtaui}) is obvious.
Assertions (\assertref{SPIassert:d1Ld2})--(\assertref{SPIassert:imaxLambdai})
follow from
(\assertref{SPIassert:degLamLLami})--(\assertref{SPIassert:diLdLtaui}),
followed by the swap in
lines~\ref{line:SPIswapbegin}--\ref{line:SPIswapend}. 

As for (\assertref{SPIassert:rdilessd}),
when $b^{(i)}(x)$ is visited for the very first time
(i.e., the first execution of line~\ref{line:SPIupdateLambda}
for some index $i$),
we have $d=\deg m^{(i)}(x)$ and 
$\rd^{(i)}(\Lambda) < d$ is obvious. 
For all later executions of line~\ref{line:SPIupdateLambda},
we have 
$d=\rd^{(i)}(\Lambda)$ and $d^{(i)}=\rd^{(i)}(\Lambda^{(i)})$
before line~\ref{line:SPIupdateLambda}, 
and $\rd^{(i)}(\Lambda) < d$ after line~\ref{line:SPIupdateLambda}
follows from Lemma~\ref{lemma:SPIAlgcore}.

To prove (\assertref{SPIassert:altLamdeginvariant}) and (\assertref{SPIassert:degLamLLami.1}),
we note that
Line~\ref{line:SPIupdateLambda} changes the degree of $\Lambda(x)$
only in iterations where
lines \ref{line:SPIswapbegin}--\ref{line:SPIresetdelta} are executed,
see (\ref{eqn:degLameisLamkplusDeltak}) below;
every later executions of Line~\ref{line:SPIupdateLambda} for the fixed $i$
does not change $\deg \Lambda(x)$ because of Lemma~\ref{lemma:SPIAlgcore} and 
that $\Lambda^{(i)}(x)$ and $d^{(i)}$ 
remain the same during the inner repeat loop.

If lines \ref{line:SPIswapbegin}--\ref{line:SPIresetdelta} are executed,
then line~\ref{line:SPIupdateLambda} changes the degree of 
$\Lambda(x)$ to
\begin{equation} \label{eqn:degLameisLamkplusDeltak}
\deg \Lambda^{(i)}(x) + d - d^{(i)} =\deg \Lambda_k(x)+\Delta_k,
\end{equation}
which is (\assertref{SPIassert:altLamdeginvariant}).
With (\assertref{SPIassert:degLambdai}), 
the left-hand side of (\ref{eqn:degLameisLamkplusDeltak}) yields also (\assertref{SPIassert:degLamLLami.1}).

It remains to prove (\assertref{SPIassert:Lamdeginvariant}).
First, we note that (\assertref{SPIassert:Lamdeginvariant})
clearly holds when the \pkw{loop} is entered for the first time.
But if (\assertref{SPIassert:Lamdeginvariant}) holds, 
then $ \Lambda^{(i)}(x)$ in (\assertref{SPIassert:degLamiLLam})
satisfies
\begin{IEEEeqnarray}{rCl}
\deg \Lambda^{(i)}(x) &=&\sum_{j\neq i}^L \big(\deg
            m^{(j)}(x)-d^{(j)}\big) \nonumber \\ &+& \deg
            m^{(i)}(x)-d.
\end{IEEEeqnarray}
It then follows from (\ref{eqn:degLameisLamkplusDeltak})
that $\Lambda(x)$ after line~\ref{line:SPIupdateLambda}
satisfies
\begin{equation}
\deg \Lambda^{(i)}(x) + d - d^{(i)} = \sum_{j=1}^L \big(\deg
m^{(j)}(x)-d^{(j)}\big),
\end{equation}
which is (\assertref{SPIassert:Lamdeginvariant}).
(Note that (\assertref{SPIassert:Lamdeginvariant}) 
and (\assertref{SPIassert:diLdLtaui}) together 
provide an alternative proof of Proposition~\ref{proposition:SolutionDegree}.)

Finally, we note that the algorithm is guaranteed to terminate because 
every execution of the \pkw{repeat} loop
(lines~\ref{line:SPIrepeat}--\ref{line:SPIuntil}) strictly
decreases $\delta_\text{max}(\Lambda)$ or $i_\text{max}(\Lambda)$
according to Lemma~\ref{lemma:SPIAlgcore}
and the swap in
lines~\ref{line:SPIswapbegin}--\ref{line:SPIswapend} strictly
decreases $d^{(i)}$.

For later use, we also record the following fact from (\extraref{SPIextra:dk})
and (\assertref{SPIassert:altLamdeginvariant}):
Let $\Lambda_1(x), \Lambda_2(x), \ldots, \Lambda_K(x)$
be all polynomials $\Lambda_k(x)$ from (\extraref{SPIextra:dk})
and let $\hat\Lambda(x)$ be the $\Lambda(x)$ returned by the algorithm.
(Note that $\deg\hat\Lambda(x)>\deg\Lambda_K(x)$.)
\begin{proposition} \label{SPIprop:degLamisincreasing}
The polynomials $\Lambda_k(x)$ defined in  (\extraref{SPIextra:dk}) satisfy
$\deg\Lambda_{1}(x)=0$ (since $\Lambda_{1}(x)=1$) 
and
\begin{equation} \label{SPIeqn:degLamisincreasing}
\deg\Lambda_{K}(x)>\ldots>\deg\Lambda_{2}(x)>\deg\Lambda_{1}(x)
\end{equation}
with
\begin{equation} \label{proof:hatLambdaj.0}
\Delta_{k}=\deg \Lambda_{k+1}(x)- \deg\Lambda_{k}(x)
\end{equation}
for $k \in \{ 1,\ldots, K-1\}$
and
\begin{equation} \label{proof:hatLambdaj.1}
\Delta_K=\deg \hat\Lambda(x)-\deg\Lambda_K(x).
\end{equation}
Moreover, we have
\begin{equation} \label{proof:deltajdjdi}
\Delta_k=d_k - \deg\!\Big( b^{(i_k)}(x)\Lambda_k(x) \bmod m^{(i_k)}(x) \Big)
\end{equation}
for $k=1,2,\ldots,K$. 
\end{proposition}

\subsection{Completing the Proof of Theorem~\ref{theorem:CorrectnessOfBasicAlg}}

It is clear at this point that the algorithm terminates 
and the returned polynomial $\Lambda(x)=\hat\Lambda(x)$
satisfies (\ref{eqn:SPI}) for all $i \in \{ 1,\ldots, L\}$.
Below, we will show that any nonzero $\tilde \Lambda(x) \in F[x]$ 
with $\deg \tilde \Lambda(x)<\deg \hat \Lambda(x)$ cannot satisfy (\ref{eqn:SPI}) for all $i$.

To this end, we need Proposition~\ref{SPIprop:degLamisincreasing} and the lemma.
\begin{lemma} \label{lemma:leadingqtLambdat} 
For any nonzero $q_{k}(x)$ with $\deg q_k<\Delta_k$, we have 
\begin{equation} \label{proof:imaxqtHatLambdat}
i_\text{max}(q_{k}\Lambda_{k})=i_\text{max}(\Lambda_{k})
\end{equation}
and
\begin{IEEEeqnarray}{rCl}
\rd^{(i_k)}(q_k \Lambda_k)& =& \deg q_k + \rd^{(i_k)}(\Lambda_k)\\
&<&d_k.
\end{IEEEeqnarray}
Moreover, for any nonzero $q_k\Lambda_k$ and $q_{k'}\Lambda_{k'}$ 
with 
\begin{equation}
i_\text{max}(\Lambda_{k})=i_\text{max}(\Lambda_{k'})
\end{equation}
(i.e., $i_k=i_{k'}$), 
we have
\begin{equation}\label{eqn:SecondClaim}
\rd^{(i_k)}(q_{k} \Lambda_{k}+q_{k'} \Lambda_{k'})=\rd^{(i_k)}(q_{k} \Lambda_{k})
\end{equation}
if $k<k'$.
\end{lemma}
\begin{proof}
First, we establish from $\deg q_k<\Delta_k$
and (\ref{proof:deltajdjdi}) that
\begin{equation}\label{proof:qjdj}
\deg q_k(x) + \deg\!\Big( b^{(i_k)}(x)\Lambda_k(x) \bmod m^{(i_k)}(x) \Big) < d_k
\end{equation}
for $k=1,2,\ldots,K$. 
For all $i\in\{1,\ldots,L\}$, we clearly have
\begin{equation}
\rd^{(i)}(q_k \Lambda_k) \leq \deg q_k + \rd^{(i)}(\Lambda_k).
\end{equation}
For $i_k$, however, we have 
\begin{equation} \label{eqn:RemDegqkLkeq}
\rd^{(i_k)}(q_k \Lambda_k) = \deg q_k + \rd^{(i_k)}(\Lambda_k)
\end{equation}
from (\ref{proof:qjdj}) and since $d_k \leq \deg m^{(i_k)}(x)$.
The first claim of the lemma then follows from $i_\text{max}(\Lambda_k)=i_k$;
the second claim (\ref{eqn:SecondClaim}) is clear from 
\begin{equation}
\rd^{(i_k)}(q_{k'} \Lambda_{k'})<d_{k'}\leq\rd^{(i_k)}(q_k \Lambda_k)<d_k
\end{equation}
for $k<k'$.
\end{proof}

Partitioning the indices $k \in \{ 1,\ldots, K\}$
into sets $S_1,\ldots,S_L$ 
such that 
\begin{equation}
k\in S_i \Longleftrightarrow i_\text{max}(\Lambda_k)=i_k=i,
\end{equation}
we obtain from Lemma~\ref{lemma:leadingqtLambdat} the following corollary.
\begin{corollary}\label{CorollaryForSPIproof}
For any $S_i$, any nonzero
\begin{equation} \label{eqn:tildeLambdaiAltform}
\tilde\Lambda^{(i)}(x)\eqdef \sum_{k\in S_i}q_k(x)\Lambda_{k}(x)
\end{equation}
with $\deg q_k<\Delta_k$ 
satisfies
\begin{equation} \label{Eqinlemma:imaxtildeLambdai}
i_\text{max}(\tilde\Lambda^{(i)})=i
\end{equation}
and
\begin{equation} \label{Eqinlemma:dmaxtildeLambdai}
\rd^{(i)}(\tilde \Lambda^{(i)})=\max_{k\in S_i}\big(q_k\Lambda_k\big)
\end{equation}
\eproofnegspace
\end{corollary}
Finally, we note that 
any nonzero $\tilde \Lambda(x) \in F[x]$ 
with $\deg \tilde \Lambda(x)<\deg \hat \Lambda(x)$ 
can be uniquely written as
\begin{equation}
\tilde \Lambda(x)=\sum_{k=1}^{K} q_k(x)\Lambda_{k}(x)
\end{equation}
for some nonzero $q_k(x)$ with 
$\deg q_k(x)< \Delta_k$.
It then follows from Corollary~\ref{CorollaryForSPIproof}
that 
$\tilde \Lambda(x)=\sum_{i=1}^L \tilde\Lambda^{(i)}(x)$ cannot satisfy (\ref{eqn:SPI}) for all $i$.

\subsection{Proof of Theorem~\ref{theorem:NumIterations}}
\label{sec:ProofOftheorem:NumIterations}

For each $i\in \{1,\ldots,L\}$, let $d^{(i)}$ be as in the algorithm, and
let $\tilde d^{(i)}$ denote the value of $d^{(i)}$ when the algorithm stops. 
Note that $d^{(i)}$ (for each $i$) is 
initialized to $\deg m^{(i)}(x)$ in line~\ref{line:SPIinitialdi}.
Now let $\delta^{(i)}\eqdef d^{(i)}-\tau^{(i)}$ for every $i\in \{1,\ldots,L\}$.
Clearly, $\delta^{(i)}$ satisfies 
\begin{equation}
\tilde d^{(i)}-\tau^{(i)}\leq \delta^{(i)}\leq \deg m^{(i)}(x)-\tau^{(i)};
\end{equation}
moreover, every execution of the swap in
lines~\ref{line:SPIswapbegin}--\ref{line:SPIswapend} strictly
reduces $\delta^{(i)}$.
Finally, let $\delta$ be as in the algorithm,  which is initialized to
\begin{equation}
\delta:=\max_{i\in\{1,\ldots,L\}}\!\big(\deg m^{(i)}(x)-\tau^{(i)}\big)
\end{equation}
(see line~\ref{line:SPIinitialdelta}).
It is obvious that the number $N_\text{it}$ of executions of line~\ref{line:SPIkappa} 
of Algorithm~\ref{alg:BasicSPIAlg} (i.e., Algorithm~\ref{alg:AnnotatedSPIAlg}) 
is equal to the total number of 
iterations of lines~\ref{line:SPIrepeat}--\ref{line:SPIuntil}.
These executions of line~\ref{line:SPIkappa} 
(with the help of line~\ref{line:SPIupdateLambda}) 
are made to make $\Lambda(x)$ satisfy (\ref{eqn:SPI}) for all $i\in \{1,\ldots,L\}$
(which holds when $\delta\leq 0$)
accompanied by the reduction of $\delta^{(i)}$ from $\deg m^{(i)}(x)-\tau^{(i)}$ 
to $\tilde d^{(i)}-\tau^{(i)}$ for every $i$. 
We therefore have
\begin{equation}
N_\text{it} = \maxD + \sum_{i=1}^L n_\text{it}^{(i)},
\end{equation}
where $\maxD$ is  defined in (\ref{def:Bbound}) and 
where $n_\text{it}^{(i)}$ denotes the number of executions of line~\ref{line:SPIkappa} 
needed for decreasing $\delta^{(i)}$ from $\deg m^{(i)}(x)-\tau^{(i)}$ to $\tilde d^{(i)}-\tau^{(i)}$.
The quantity $n_\text{it}^{(i)}$ for each $i\in\{1,\ldots, L\}$ 
is 
\begin{IEEEeqnarray}{rCl}
n_\text{it}^{(i)}&=&L\cdot \big(\deg m^{(i)}(x)-\tau^{(i)}-(\tilde d^{(i)}-\tau^{(i)})\big)\\
&=&L\cdot \big(\deg m^{(i)}(x)-\tilde d^{(i)}\big)
\end{IEEEeqnarray}
and thus
\begin{equation}
\sum_{i=1}^L n_\text{it}^{(i)}=L \cdot \sum_{i=1}^L\big(\deg m^{(i)}(x)-\tilde d^{(i)}\big).
\end{equation}
We therefore obtain
\begin{equation}\label{Nit:ProofExactBound}
N_\text{it} = \maxD + L \cdot \sum_{i=1}^L\big(\deg m^{(i)}(x)-\tilde d^{(i)}\big).
\end{equation}
But
$\sum_{i=1}^L\big(\deg m^{(i)}(x)-\tilde d^{(i)}\big)=\deg \Lambda(x)$
from (\assertref{SPIassert:Lamdeginvariant}) with $d^{(i)}=\tilde d^{(i)}$.

\section{Quotient Saving Algorithm\\ and Remainder Saving Algorithm}
\label{section:QuotientRemainderSavingAlgs}

For $L=1$, the reverse Berlekamp--Massey algorithm is easily 
translated into two other algorithms, one of which is a Euclidean algorithm 
\cite{YuLoeliger2016IT}. 
In fact, it is a main point of \cite{YuLoeliger2016IT} that
these algorithms may be viewed as different versions of a single algorithm.
We now demonstrate that this works also for $L>1$. 

However, for $L>1$, these other algorithms are less attractive 
than the monomial-SPI reverse Berlekamp--Massey algorithm 
(Algorithm~\ref{alg:BasicSPIASpec})
as will be detailed below. 
However, before discounting these other algorithms from future research, 
it may be remembered that the complexity of the asymptotically fast MLFSR algorithms of 
\cite{SiBo:fskew2014} and \cite{Nielsen2013} is cubic in $L$
while the complexity of the algorithms below is only quadratic in $L$.

\subsection{Quotient Saving Algorithm}

Algorithm~\ref{alg:QSSPIAlg} (see box) 
is a variation of Algorithm~\ref{alg:BasicSPIAlg} 
that achieves a generalization 
of Algorithm~\ref{alg:BasicSPIASpec} to general $m^{(i)}(x)$. 
To this end, we store and update the quotients 
$Q^{(i)}(x)$, $i=1,\ldots,L$, defined by
\begin{equation} \label{SPIeqn:quitientandremainder}
b^{(i)}(x)\Lambda(x)=Q^{(i)}(x)m^{(i)}(x)+r^{(i)}(x)
\end{equation}
with $r^{(i)}(x)\eqdef b^{(i)}(x)\Lambda(x)\bmod m^{(i)}(x)$. The coefficient of 
$x^d$ of $r^{(i)}(x)$ in line~\ref{line:SPIkappa} of Algorithm~\ref{alg:BasicSPIAlg} 
can then be computed as in line~\ref{line:QSSPIkappa} of Algorithm~\ref{alg:QSSPIAlg}.

The quotients $Q^{(i)}(x)$ of (\ref{SPIeqn:quitientandremainder}) are initialized 
in line~\ref{line:QSSPIinitialQj}, updated in line~\ref{line:QSSPIupdateQj}, 
and stored (as $Q^{(i,j)}(x)$) in line~\ref{line:QSSPIstoreQj}, in parallel with $\Lambda(x)$.

All other quantities in the algorithm remain unchanged.
In any case (as in (\ref{algorithm:rdiLambda})), we always have
\begin{equation}
\deg\!\big(b^{(i)}(x)\Lambda(x)-Q^{(i)}(x)m^{(i)}(x)\big)<d
\end{equation}
after executing lines~\ref{line:QSSPIupdateQbgn}--\ref{line:QSSPIupdateQend}.

Theorem~\ref{theorem:NumIterations} still applies,
with ``line~\ref{line:SPIkappa}'' replaced by ``line~\ref{line:QSSPIkappa}''.
Due to the additional computation of lines~\ref{line:QSSPIupdateQbgn}--\ref{line:QSSPIupdateQend},
the complexity of Algorithm~\ref{alg:QSSPIAlg} 
is 
\begin{equation} \label{eqn:QuotientSavingComplexity}
O\big (N_\text{it} L \deg \Lambda(x) \big) \leq O\big( L(\maxD D + L D^2) \big).
\end{equation}

Compared with (\ref{eqn:MonomialRevBMAComplexity}),
the factor $L$ in (\ref{eqn:QuotientSavingComplexity}) 
makes 
this algorithm
less attractive for $L>1$
than Algorithm~\ref{alg:BasicSPIASpec}.

\begin{table}[tp]
\framebox[\linewidth]{%
\normalsize%
\begin{minipage}{0.95\linewidth}
\begin{algorithm}[Quotient Saving SPI Algorithm]\label{alg:QSSPIAlg}\\
{Input:} $b^{(i)}(x), m^{(i)}(x), \tau^{(i)}$ for $i=1,\ldots,L$.\\
{Output:} $\Lambda(x)$ as in the problem statement.
\begin{pseudocode}
\npcl[line:QSSPIforinitial] \pkw{for} $i=1,\ldots, L$ \pkw{begin}\\
\npcl \>$\Lambda^{(i)}(x):=0$\\
\npcl \>$d^{(i)}:=\deg m^{(i)}(x)$\\
\npcl[line:QSSPIinitialkappa] \>$\kappa^{(i)}:=\lcf m^{(i)}(x)$\\
\npcl[line:QSSPIinitialQijbgn]\> \pkw{for} $j=1,\ldots, L$ \pkw{begin} \\
\npcl\>\> $Q^{(i,j)}(x):=0$\\
\npcl\>\> \pkw{if} $i=j$ \pkw{begin} $Q^{(i,j)}(x):=-1$ \pkw{end}\\
\npcl[line:QSSPIforinitialend]\> \pkw{end} \\
\npcl[line:QSSPIforendinitial] \pkw{end} \\
\npcl[line:QSSPIinitialLam] $\Lambda(x):=1$\\
\npcl[line:QSSPIinitialQj]\pkw{for} $i=1,\ldots, L$ \pkw{begin} $Q^{(i)}(x):=0$ \pkw{end} \\
\npcl[line:QSSPIinitialdelta] $\delta:=\max_{i\in\{1,\ldots,L\}}\big(\deg m^{(i)}(x)-\tau^{(i)}\big)$\\
\npcl[line:QSSPIinitialiddx] $i:=1$\\
\npcl[line:QSSPIloopbegin]   \pkw{loop begin} \\
\npcl[line:QSSPIrepeat] \> \pkw{repeat}\\
\npcl[line:QSSPIupdatei]\> \>  \pkw{if} $i>1$  \pkw{begin} $i:=i-1$ \pkw{end}\\
\npcl[line:QSSPIupdtatedelta]\> \>  \pkw{else begin} \\
\npcl[line:QSSPIstop] \>\>\> \pkw{if} $\delta\leq 0$  \pkw{return} $\Lambda(x)$\\
\npcl \>\>\>$i:=L$\\
\npcl \>\>\>$\delta:=\delta-1$\\
\npcl[line:QSSPIupdateiEnd]\> \> \pkw{end} \\
\npcl[line:QSSPIupdated]\>\> $d:=\delta+\tau^{(i)}$\\
\npcl[line:QSSPIkappa] \> \> 
    $\displaystyle \kappa := 
    \sum_{\ell} b_{d-\ell}^{(i)}\Lambda_\ell
  - \sum_{\ell} m_{d-\ell}^{(i)}Q_\ell^{(i)}$ \\
\npcl[line:QSSPIuntil] \> \pkw{until} $\kappa \neq 0$\\
\> \> \> {\grey\rule[0.5ex]{50mm}{1pt}}\\
\npcl[line:QSSPIifswap] \> \pkw{if} $ d<d^{(i)}$ \pkw{begin} \\
\npcl[line:QSSPIswapbegin] \>\> \pkw{swap} $(\Lambda(x),\Lambda^{(i)}(x))$\\
\npcl[line:QSSPIswapddi] \>\> \pkw{swap} $(d, d^{(i)})$\\
\npcl[line:QSSPIswapend] \>\> \pkw{swap} $(\kappa, \kappa^{(i)})$\\
\npcl[line:QSSPIstoreQj]\>\>\pkw{for} $j=1,\ldots, L$ \pkw{swap} $(Q^{(j)}(x), Q^{(i,j)}(x))$ \\
\npcl[line:QSSPIresetdelta] \>\> $\delta:=d-\tau^{(i)}$ \\
\npcl[line:QSSPIifswapend] \> \pkw{end} \\
\> \> \> {\grey\rule[0.5ex]{50mm}{1pt}}\\
\npcl[line:QSSPIupdateLambda] \>
$\Lambda(x):= \kappa^{(i)} \Lambda(x)- \kappa x^{d-d^{(i)}} \Lambda^{(i)}(x)$ \\
\npcl[line:QSSPIupdateQbgn]\> \pkw{for} $j=1,\ldots, L$ \pkw{begin}\\
\npcl[line:QSSPIupdateQj]\>\> $Q^{(j)}(x):= \kappa^{(i)}Q^{(j)}(x)- \kappa x^{d-d^{(i)}} Q^{(i,j)}(x)$ \\
\npcl[line:QSSPIupdateQend]\> \pkw{end}\\
\npcl[line:QSSPIloopend]   \pkw{end}
\end{pseudocode}
\end{algorithm}
\end{minipage}%
}
\end{table}

\begin{table}[tp]
\framebox[\linewidth]{%
\normalsize%
\begin{minipage}{0.95\linewidth}
\begin{algorithm}[Remainder Saving SPI Algorithm]\label{alg:RSSPIAlg}\\
(a Euclidean algorithm)\\[0.5ex]
{Input:} $b^{(i)}(x), m^{(i)}(x), \tau^{(i)}$ for $i=1,\ldots,L$.\\
{Output:} $\Lambda(x)$ as in the problem statement.
\begin{pseudocode}
\npcl[line:RSSPIforinitial] \pkw{for} $i=1,\ldots, L$ \pkw{begin}\\
\npcl \>$\Lambda^{(i)}(x):=0$\\
\npcl \>$d^{(i)}:=\deg m^{(i)}(x)$\\
\npcl \>$\kappa^{(i)}:=\lcf m^{(i)}(x)$\\
\npcl\> \pkw{for} $j=1,\ldots, L$ \pkw{begin} \\
\npcl\>\> $r^{(i,j)}(x):=0$\\
\npcl\>\> \pkw{if} $i=j$ \pkw{begin} $r^{(i,j)}(x):=m^{(i)}(x)$ \pkw{end}\\
\npcl\> \pkw{end} \\
\npcl[line:RSSPIforendinitial] \pkw{end} \\
\npcl[line:RSSPIinitialLam] $\Lambda(x):=1$\\
\npcl\pkw{for} $i=1,\ldots, L$ \pkw{begin} $r^{(i)}(x):=b^{(i)}(x)$ \pkw{end} \\
\npcl[line:RSSPIinitialdelta] $\delta:=\max_{i\in\{1,\ldots,L\}}\big(\deg m^{(i)}(x)-\tau^{(i)}\big)$\\
\npcl[line:RSSPIinitialiddx] $i:=1$\\
\npcl[line:RSSPIloopbegin]   \pkw{loop begin} \\
\npcl[line:RSSPIrepeat] \> \pkw{repeat}\\
\npcl[line:RSSPIupdatei]\> \>  \pkw{if} $i>1$  \pkw{begin} $i:=i-1$ \pkw{end}\\
\npcl[line:RSSPIupdtatedelta]\> \>  \pkw{else begin} \\
\npcl[line:RSSPIstop] \>\>\> \pkw{if} $\delta\leq 0$  \pkw{return} $\Lambda(x)$\\
\npcl \>\>\>$i:=L$\\
\npcl \>\>\>$\delta:=\delta-1$\\
\npcl[line:RSSPIupdateiEnd]\> \> \pkw{end} \\
\npcl[line:RSSPIupdated]\>\> $d:=\delta+\tau^{(i)}$\\
\npcl[line:RSSPIkappa] \> \> $\kappa:= \text{coefficient of $x^{d}$ of $r^{(i)}(x)$}$ \\
\npcl[line:RSSPIuntil] \> \pkw{until} $\kappa \neq 0$\\
\> \> \> {\grey\rule[0.5ex]{50mm}{1pt}}\\
\npcl[line:RSSPIifswap] \> \pkw{if} $ d<d^{(i)}$ \pkw{begin} \\
\npcl[line:RSSPIswapbegin] \>\> \pkw{swap} $(\Lambda(x),\Lambda^{(i)}(x))$\\
\npcl[line:RSSPIswapddi] \>\> \pkw{swap} $(d, d^{(i)})$\\
\npcl[line:RSSPIswapend] \>\> \pkw{swap} $(\kappa, \kappa^{(i)})$\\
\npcl\>\>\pkw{for} $j=1,\ldots, L$ \pkw{swap} $(r^{(j)}(x), r^{(i,j)}(x))$ \\
\npcl[line:RSSPIresetdelta] \>\> $\delta:=d-\tau^{(i)}$ \\
\npcl[line:RSSPIifswapend] \> \pkw{end} \\
\> \> \> {\grey\rule[0.5ex]{50mm}{1pt}}\\
\npcl[line:RSSPIupdateLambda] \>
$\Lambda(x):= \kappa^{(i)} \Lambda(x)- \kappa x^{d-d^{(i)}} \Lambda^{(i)}(x)$ \\
\npcl[line:RSSPIupdateRbgn]\> \pkw{for} $j=1,\ldots, L$ \pkw{begin}\\
\npcl[line:RSSPIupdateRj]\>\> $r^{(j)}(x):= \kappa^{(i)}r^{(j)}(x)- \kappa x^{d-d^{(i)}} r^{(i,j)}(x)$ \\
\npcl[line:RSSPIupdateRend]\> \pkw{end}\\
\npcl[line:RSSPIloopend]   \pkw{end}
\end{pseudocode}
\end{algorithm}
\end{minipage}%
}
\end{table}

\subsection{Remainder Saving Algorithm}

Another variation of Algorithm~\ref{alg:BasicSPIAlg} is Algorithm~\ref{alg:RSSPIAlg}, 
where we store and update the remainders $r^{(i)}(x)$
from (\ref{SPIeqn:quitientandremainder}). 
In consequence, the computation 
of line~\ref{line:SPIkappa} of Algorithm~\ref{alg:BasicSPIAlg} 
is unnecessary and replaced by the trivial 
line~\ref{line:RSSPIkappa} of Algorithm~\ref{alg:RSSPIAlg}. 
However, updating the remainders $r^{(i)}(x)$
requires the additional computation in 
lines~\ref{line:RSSPIupdateRbgn}--\ref{line:RSSPIupdateRend}. 

Otherwise, the algorithm works exactly like Algorithms~\ref{alg:BasicSPIAlg}~and~\ref{alg:QSSPIAlg}. 
In particular, we always have
\begin{equation}
\deg r^{(i)}(x)<d
\end{equation}
after executing lines~\ref{line:RSSPIupdateRbgn}--\ref{line:RSSPIupdateRend}.

Note that this algorithm is rather 
a Euclidean algorithm 
\cite{FengTzeng1989,GathenGerhard} 
than a Berlekamp--Massey algorithm
(but it is a new algorithm as well).

Due to the computation of 
lines~\ref{line:RSSPIupdateRbgn}--\ref{line:RSSPIupdateRend},
the complexity of Algorithm~\ref{alg:RSSPIAlg} is 
\begin{equation} \label{eqn:RemainderSavingComplexity}
O(N_\text{it}L \nu_\text{max})
\end{equation}
with 
\begin{equation}
\nu_\text{max} \eqdef \max_{i \in \{ 1,\ldots,L\} } \deg m^{(i)}(x).
\end{equation}
But
\begin{equation} \label{eqn:RemainderSavingComplexityComparison}
N_\text{it} L \nu_\text{max} \geq N_\text{it} \deg \Lambda(x)
\end{equation}
by (\ref{eqn:DegreeBoundLCM}).
It follows that 
the complexity of Algorithm~\ref{alg:RSSPIAlg} is never smaller 
than the complexity (\ref{eqn:MonomialRevBMAComplexity})
of Algorithm~\ref{alg:BasicSPIASpec}. 
In particular, for monomial moduli, we have 
$\deg \Lambda(x) \leq \nu_\text{max}$ (by Proposition~\ref{proposition:MonomialDegreeBound}), 
i.e., the left side of 
(\ref{eqn:RemainderSavingComplexityComparison}) exceeds the right side by a factor of $L$.

For general SPI problems, the difference between the left side and the right side
of (\ref{eqn:RemainderSavingComplexityComparison})
may be small, in which case Algorithm~\ref{alg:RSSPIAlg}
may be more attractive than first monomializing the SPI problem
and then using Algorithm~\ref{alg:BasicSPIASpec}.
Indeed, (\ref{eqn:RemainderSavingComplexityComparison}) is an equality 
if $\deg \Lambda(x) = L \nu_\text{max}$, which happens
in the very special case where 
\begin{equation}
\lcm \big\{ m^{(1)}(x),\ldots, m^{(L)}(x) \big\} = m^{(1)}(x) \cdots m^{(L)}(x),
\end{equation}
$\deg m^{(1)}(x) = \ldots = \deg m^{(L)}(x) = \nu_\text{max}$,
and $\tau^{(1)} = \ldots = \tau^{(L)} = 0$.
However, this case does not arise in decoding as in this paper.

\section*{Acknowledgment}

The paper has much benefitted from the comments by the anonymous reviewers.

\newcommand{\COM}{IEEE Trans.\ Communications}
\newcommand{\COMMag}{IEEE Communications Mag.}
\newcommand{\IT}{IEEE Trans.\ Inf.\ Theory}
\newcommand{\JSAC}{IEEE J.\ Select.\ Areas in Communications}
\newcommand{\SP}{IEEE Trans.\ Signal Proc.}
\newcommand{\SPMag}{IEEE Signal Proc.\ Mag.}
\newcommand{\ProcIEEE}{Proceedings of the IEEE}


\vspace{-12ex}

\begin{IEEEbiographynophoto}{Jiun-Hung~Yu}(S'10-M'14)
was born in Nantou, Taiwan, in 1979.
He received the M.S. degree in communication engineering from 
National Chiao Tung University, Hsinchu, Taiwan, in 2003, and
the Ph.D.\ degree in electrical engineering from ETH Zurich, in 2014.
From 2003 to 2008, he was with Realtek Semiconductor Cooperation,
Hsinchu, Taiwan. From 2008 to 2017, he was with the Signal and Information
Processing Laboratory of ETH Zurich.
Since 2017, he has been with National Chiao Tung University. He is currently 
an Assistant Professor.  
His research interests include communication theory, error-correcting codes,
and statistical signal processing.
\end{IEEEbiographynophoto}

\vspace{-12ex}

\begin{IEEEbiographynophoto}{Hans-Andrea Loeliger} (S'85-M'92-SM'03-F'04) 
received the Diploma in electrical engineering and the Ph.D.\ degree in 1992 from ETH Zurich, Switzerland. 
From 1992 to 1995, he was with Link\"oping University, Link\"oping Sweden. 
From 1995 to 2000, he was a technical consultant and coowner of a consulting company. 
Since 2000, he has been a Professor with the Department of Information Technology 
and Electrical Engineering of ETH Zurich, Switzerland. 
His research interests include the broad areas of signal processing, 
machine learning, information theory, error correcting codes, communications, and electronic circuits.
\end{IEEEbiographynophoto}

\end{document}